\pdfoutput=1

\documentclass[11pt]{article}

\usepackage[review]{acl}

\usepackage{times}
\usepackage{latexsym}

\usepackage[T1]{fontenc}

\usepackage[utf8]{inputenc}

\usepackage{microtype}

\usepackage{inconsolata}

\usepackage{graphicx}

\usepackage{microtype}
\microtypecontext{spacing=nonfrench}

\usepackage{hyperref}
\usepackage{longtable}
\usepackage{arydshln}
\usepackage{bm}     
\usepackage{booktabs}
\usepackage{array}
\usepackage{makecell}
\usepackage{amsmath, amssymb, amsthm}
\usepackage{cleveref}  
\usepackage{amsmath, etoolbox}

\makeatletter
\patchcmd{\gathered@d@t@}
  {\tagform@{\thetag@form@}\fi}
  {%
    \setbox\z@\hbox{\tagform@{\thetag@form@}}%
    \dimen@=\ht\z@ \advance\dimen@ by \dp\z@
    \rlap{\raisebox{0.5\dimen@}{\tagform@{\thetag@form@}}}%
  }{}{}
\makeatother
\newtheorem{definition}{Definition}
\newtheorem{theorem}{Theorem}

\usepackage{algorithm}
\usepackage{algpseudocode}
\usepackage{ragged2e} 
\usepackage{graphicx}
\usepackage{xspace}
\usepackage{multirow}
\usepackage{pifont} 
\usepackage{colortbl}
\PassOptionsToPackage{table}{xcolor} 
\usepackage{xcolor}
\usepackage{comment}
\newcommand{\modelname}{\textsc{Greater}\xspace}
\newcommand{\attackname}{\textsc{Greater-A}\xspace}
\newcommand{\defensename}{\textsc{Greater-D}\xspace}

\newcommand{\eg}{\emph{e.g.},\xspace}

\newcommand{\etc}{etc.\xspace}

\newcommand\secref[1]{\S\ref{#1}}

\newcommand{\fakeparagraph}[1]
{\vspace{1mm}\noindent\textbf{#1}}

\title{\Large Iron Sharpens Iron: Defending Against Attacks in Machine-Generated Text Detection with Adversarial Training}

\author{Yuanfan Li\textsuperscript{$1,\dagger$}, Zhaohan Zhang\textsuperscript{$2,\dagger$}, Chengzhengxu Li\textsuperscript{$1,\dagger$},  Chao Shen\textsuperscript{$1$}, Xiaoming Liu\textsuperscript{$1,\ast$} \\
        \textsuperscript{1}Faculty of Electronic and Information Engineering, Xi'an Jiaotong University\\ 
        \textsuperscript{2}Queen Mary University of London \\
        \textsuperscript{$\dagger$} Equal contribution, \textsuperscript{$\ast$} Corresponding author\\
        \texttt{
        liyuan7716@gmail.com, czx.li@stu.xjtu.edu.cn
        }\\
        \texttt{
        zhaohan.zhang@qmul.ac.uk, \{chaoshen, xm.liu\}@xjtu.edu.cn 
        }\\
        }

\begin{document}
\maketitle
\begin{abstract}
Machine-generated Text (MGT) detection is crucial for regulating and attributing online texts.
While the existing MGT detectors achieve strong performance, they remain vulnerable to simple perturbations and adversarial attacks.
To build an effective defense against malicious perturbations, we view MGT detection from a threat modeling perspective, that is, analyzing the model's vulnerability from an adversary's point of view and exploring effective mitigations. 
To this end, we introduce an adversarial framework for training a robust MGT detector, named \textbf{GRE}edy \textbf{A}dversary
Promo\textbf{T}ed
Defend\textbf{ER}
(\modelname).
The \modelname consists of two key components: an adversary \attackname and a detector \defensename.
The \defensename learns to defend against the adversarial attack from \attackname and generalizes the defense to other attacks.
\attackname identifies and perturbs the critical tokens in embedding space, along with greedy search and pruning to generate stealthy and disruptive adversarial examples. 
Besides, we update the \attackname and \defensename synchronously, encouraging the \defensename to generalize its defense to different attacks and varying attack intensities. 
Our experimental results across 10 text perturbation strategies and 6 adversarial attacks show that our \defensename reduces the Attack Success Rate (ASR) by \textbf{0.67\%} compared with SOTA defense methods while our \attackname is demonstrated to be more effective and efficient than SOTA attack approaches. 
Codes and dataset are available in \url{https://github.com/Liyuuuu111/GREATER}.
\end{abstract}

\section{Introduction}
\label{sec:intro}
The rapid development of large language models (LLM)~\cite{achiam2023gpt,dubey2024llama,claude3, guo2025deepseek} enables the model to generate highly human-like texts, which has raised broad concerns about the unrestricted dissemination of non-attributed textual contents including misinformation, fabricate news, and phishing emails.
These negative impacts of MGTs lead to extensive works on MGT detection \cite{mitchell2023detectgpt, verma2024ghostbuster, liu2024does, baofast, liu2024does} to accurately attribute the authorship of textual content and inform the readers.
\begin{figure}[t]
    \centering
    \resizebox{0.5\textwidth}{!}{%
        \includegraphics{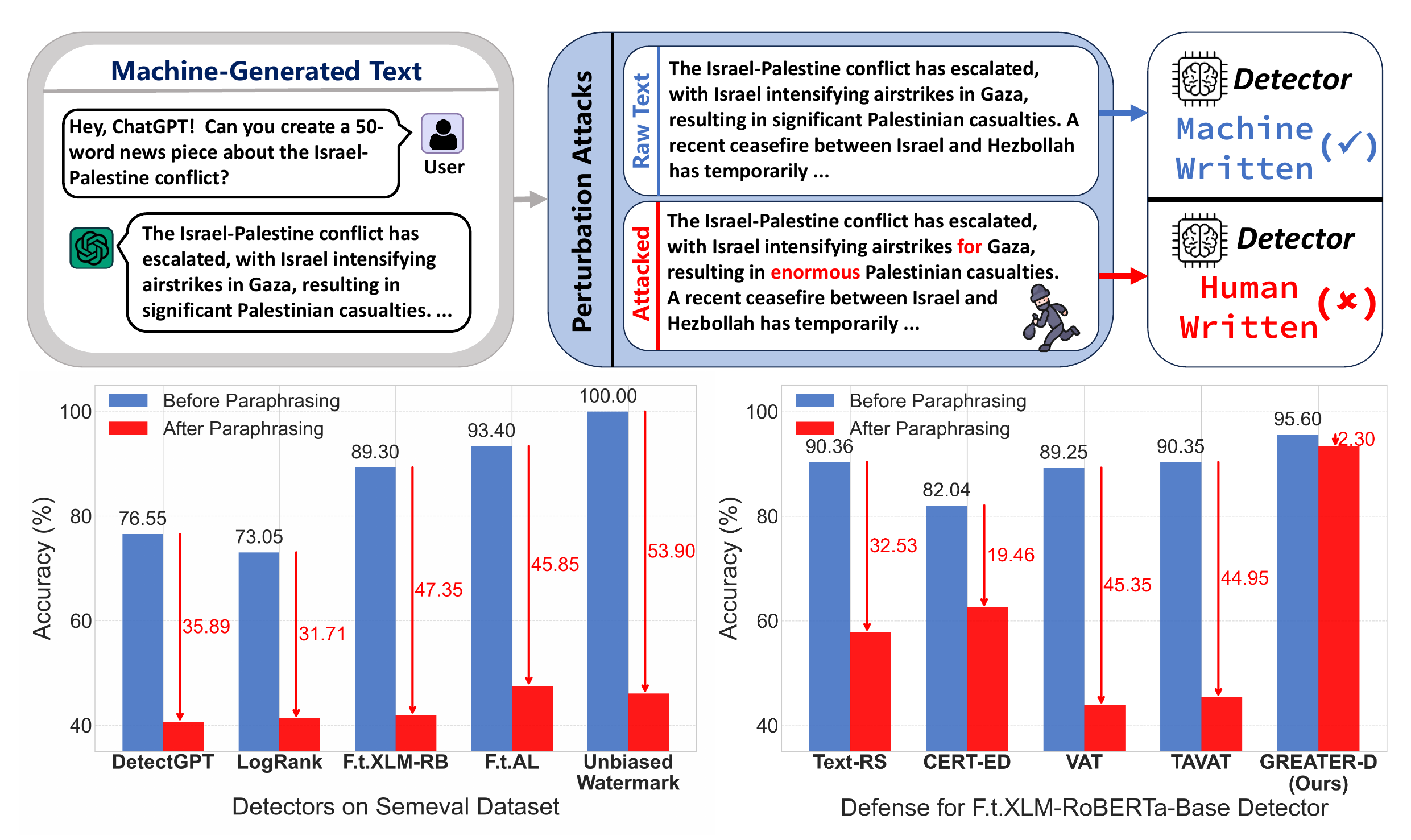} 
    }
	\caption{Performance drop of different MGT detectors and defense methods under text perturbation\protect\footnotemark.}
    \vspace{-0.4cm}
    
    \label{fig:compare1} 
\end{figure}

\footnotetext{The detectors include 
    DetectGPT~\cite{mitchell2023detectgpt}, LogRank~\cite{su2023detectllm}, fine-tuned xlm-roberta-base~\cite{liu2019roberta}, fine-tuned Albert-large~\cite{lan2019albert}, and Unbiased Watermark~\cite{hu2023unbiased}.
    The defense methods contain Text-RS~\cite{zhang2024random}, CERT-ED~\cite{huang2024cert}, VAT~\cite{miyato2016virtual}, TAVAT~\cite{li2021tavat}, and GREATER-D (ours).}

Despite the superb performance of current MGT detectors, a recent study \cite{wang2024stumbling} finds an astonishing fact that \textit{all} detectors exhibit different loopholes in robustness, that is, existing detectors suffer great performance drop when facing different \textbf{text perturbation strategies} including editing \cite{kukich1992techniques, gabrilovich2002homograph}, paraphrasing \cite{shi2024red}, prompting \cite{zamfirescu2023johnny}, and co-generating \cite{kushnareva2024ai}, \etc
As illustrated in Figure~\ref{fig:compare1}, the detection accuracy of current detectors drops by around 30\%-50\% when confronted with simple perturbations, and the defense methods for general text classification cannot be simply adapted to the MGT detection scenario.
More seriously, the vulnerability of MGT detectors is also unveiled by \textbf{adversarial attacks} that exploit the internal state~\cite{yoo2021towards} or outputs~\cite{liu2024hqa,yu2024query,hu2024fasttextdodger} of the detectors through multiple queries. 
Alas, there are few works on improving the robustness against adversarial attacks for MGT detectors.

\noindent\textbf{Motivation.} 
We rely on \textit{threat modeling} to
advance the robustness of the MGT detectors against perturbation and adversarial attacks.
As the proverb
says 'Iron sharpens iron', we focus on constructing powerful adversarial examples which mislead the prediction of the detector to facilitate the post-training of MGT detectors and defend against different attacks. 
Existing text perturbation strategies \cite{wang2024stumbling} adjust token distribution without accessing
information from the target MGT detector, resulting in low-quality and non-targeted adversarial
examples. 
The adversarial attacks
are only effective in white-box setting \cite{yoo2021towards} or require
excessive queries to target detectors \cite{hu2024fasttextdodger,yu2024query,liu2024hqa}.
Moreover, \citet{wang2024comprehensive} find that the defense built by adversarial training cannot generalize well to the attacks on which it was not originally trained.
To overcome these limitations, we propose an efficient adversarial training framework that works in a black-box setting and builds generalizable defense against a wide variety of perturbations and attacks for MGT detectors.

\noindent \textbf{Our Work.} 
In this paper, we propose an adversarial framework for training robust MGT detector, namely \textbf{GRE}edy \textbf{A}dversary
Promo\textbf{T}ed
Defende\textbf{R}.
(\modelname). 
\modelname consists of an adversary (\attackname) and a detector (\defensename).
The \defensename learns to discern MGTs from the human-written texts (HWTs), while the \attackname, which queries the output of the detector, aims to imply minimum perturbation on MGTs to deceive the detector.
Restricted by the scenario where \textbf{only outputs} from the target detector are available, we use an open-sourced surrogate model to retrieve gradient information to identify important tokens in the prediction. 
Afterwards, we introduce a gradient ascent perturbation on the embedding of MGTs from the surrogate model to enhance both the quality and stealthiness of generated adversarial text. 
To reduce the number of queries needed for building effective adversarial examples, we design a greedy search and pruning strategy. 
In the training stage, we update the \attackname and \defensename in the same training step so that \defensename learns from a curriculum of adversarial examples to generalize its defense.
The experiment results demonstrate that our method achieves an average ASR of 5.53\% against various attacks, which is 0.67\% lower compared to the SOTA defense methods.
We also find that \attackname achieves the most effective attack, achieving an ASR of 96.58\%, which surpasses SOTA attack methods by 8.45\% while requiring 4 times fewer queries.

Our contributions are as follows: 
\begin{itemize}
    \item \textbf{Adversarial Training Framework.} 
    We propose an adversarial training framework \modelname 
    to improve the robustness of MGT detectors, in which the adversary maliciously perturbs the MGTs to construct hard adversarial examples, while the detector is trained to maintain correct prediction towards the adversarial examples.
    We update the detector and the adversary generator in the same training step for better generalization on defense.
    \item 
    \textbf{Adversarial Examples Generation.}
    We propose a strong and efficient adversarial examples construction method in the black-box setting. 
    We retrieve gradient information from a surrogate model to rank the important tokens in MGT detection and design a greedy search and pruning strategy to reduce the query times needed for adversarial attacks.
    \item 
    \textbf{Outstanding Performance.} Testing results across 16 attack methods demonstrate that our detector outperforms 10 existing SOTA defense methods in robustness, while our adversary achieves significant improvements in both attack efficiency and effectiveness compared to 13 SOTA attack approaches. 
\end{itemize}

\section{Related Work}

\noindent \textbf{Machine-Generated Text (MGT) Detection.} 
There have been attempts to detect and attribute the MGTs in the pre-LLM era \cite{zhong2020neural, uchendu2020authorship}.
Nowadays, many works \cite{mitchell2023detectgpt, wang2023seqxgpt, liu2022coco, kushnareva2024ai, guodetective} aim to accurately annotate online texts as LLMs' astonishing ability to generate fluent, logical, and human-like content, which helps the proliferation of unchecked information.
Despite the achievements made in MGT detection, some works indicate the MGT detectors are vulnerable to simple perturbation or adversarial attacks.
For example, ~\citet{wang2024stumbling} test the robustness of eight MGT detectors with twelve perturbation strategies and they surprisingly find that none of the existing detectors remain robust under all
the attacks.
Moreover, MGT detectors' defense against adversarial attacks is also questioned \cite{fishchuk2023adversarial}.
Other studies also reveal the fact that MGT detectors suffer from authorship obfuscation \cite{macko2024authorship} and biased decision \cite{liang2023gpt}.
To mitigate the vulnerability of MGT detectors, our work focuses on improving detector robustness against text perturbations and adversarial attacks.

\noindent \textbf{Adversarial Training.} 
Adversarial training aims to optimize the model toward maintaining correct predictions against adversarial examples that are misleading data constructed for malicious purposes.
Earlier works first augment the training set with adversarial examples for defense against specific attacks~\cite{huang2024cert, zeng2023certified}.
These methods are shown to be hard to generalize to unseen attacks~\cite{wang2024comprehensive}.
\citet{yoo2021towards} update the adversary and the target model in the same step to generalize the defense to unseen attacks.
However, they rely on the availability of explicit first-order gradients, which is not applicable in the real-world case.
OUTFOX~\cite{koike2024outfox} identifies adversarial attack with in-context learning, which requires the demonstrations of adversarial samples as prompts. RADAR~\cite{hu2023radar} uses a paraphraser as the adversary. Both OUTFOX and RADAR are designed to defend against known attacks. 
Different from previous works, we propose an effective adversarial training framework for MGT detectors that builds a generalizable defense against a variety of attacks in the black-box setting.

\section{Threat Model}
We follow the standard threat modeling framework outlined in prior work \cite{biggio2018wild} and describe our assumptions about the adversary's goal and adversary's capability.

\noindent\textbf{Adversary's Goal.}\label{goal}
Given a piece of MGT, the goal of the adversary is to make trivial changes to the original MGT so as to mislead the prediction of the detector.
We refer the changed texts as \textit{adversarial examples}.
Ideally, the adversarial examples should satisfy three requirements:
\textit{i) Low Perturbation Rate.} 
Only trivial changes should be applied on the adversarial examples and the semantics of original texts should be retained.
\textit{ii) High Readability.}
Adversarial examples should exhibit high readability so that the attack is most invisible to humans.
\textit{iii) Less Query Requirements.}
The adversary should be query-efficient to reduce the time and budget needed to construct each adversarial example.
    
\noindent\textbf{Adversary's Capability. }
We assume the adversary's capability in a real-world setting.
First, an adversary only maliciously edits the MGTs but would not make any changes to the HWTs.
This is because HWTs are trustworthy and there is no need for the adversary to change the prediction on HWTs.
Second, since most commercial MGT detectors (\eg GPT Zero\footnote{\url{https://gptzero.me/}}) are close-sourced, the adversary should not have access to model weights and internal states of the target model.
The only information the adversary is permitted to query is the output of the detector.
Third, the adversary is allowed to access any open-sourced models. 

\section{Methodology} \label{sec:generator}
We introduce the framework of \modelname in this section.
The architecture of \modelname is shown in Figure \ref{fig:framework}. 
In the following subsections, we first describe the workflow of the adversary and detector, respectively. 
Then we systematically outline the adversarial training process.

\begin{figure*}[htbp]
    \centering
    \resizebox{1\textwidth}{!}{%
        \includegraphics{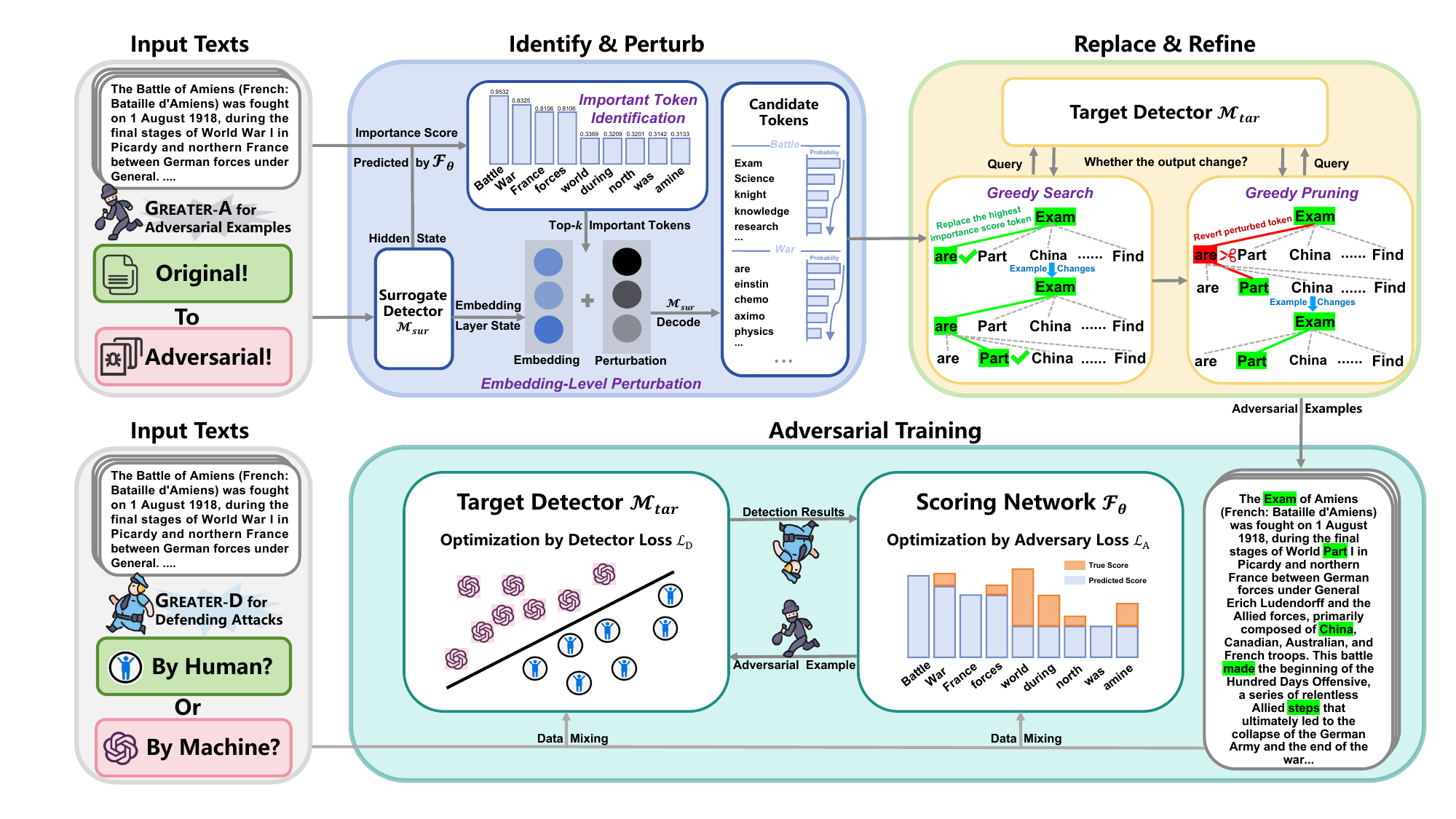} 
    }

    \caption{\textbf{Pipeline of \modelname.}
    The adversary identifies important tokens in the original MGT and generates candidates for important tokens (\secref{AEG}). The adversary conducts and refines the attack by greedy search and pruning (\secref{AEE}). The final adversarial examples are feeded to the target detector (\secref{DA})
    and participate in the adversarial training process (\secref{AT}).}
    \vspace{-0.4cm}
    \label{fig:framework} 
\end{figure*}

\subsection{\attackname for Generating Adversarial Examples}
\label{sec:advgen}
To achieve the Adversary's Goal outlined in \secref{goal}, we developed an effective and efficient adversary. Specifically, the adversary achieves these requirements through two stages: \textit{Identify \& Perturb} and \textit{Replace \& Refine}.

\subsubsection{Identify \& Perturb}\label{AEG}
In this module, we design a token importance estimation module and apply a targeted perturbation on the embeddings of important tokens.

\fakeparagraph{Important Token Identification.}
We consider a black-box setting where the internal state of the target detector $\mathcal{M}_{tar}(.)$ is inaccessible.
Given an original MGT $X = [x_1, x_2, \ldots, x_T]$ consisting of $T$ tokens, we utilize a surrogate model $\mathcal{M}_{sur}(.)$ instead to obtain the last layer hidden state of each token in the text:
\begin{equation}
\label{eq:1}
H = \mathcal{M}_{sur}(X) = [\mathbf{h}_1, \mathbf{h}_2, \ldots, \mathbf{h}_T],
\end{equation}
where $\mathbf{h}_t$ represents the last layer hidden state of the $t$-th token $x_t$ generated by $\mathcal{M}_{sur}(.)$.
To obtain the importance score of each token $s_t$, we train a simple scoring network $\mathcal{F}_{\theta}(.)$ which takes the feature embeddings as input and outputs the prediction of importance scores for each token:
\begin{equation}
\label{eq:2}
s_t = \mathcal{F}_{\theta}(\mathbf{h}_t),
\end{equation}
where $\theta$ are learnable parameters.
Then, we select the top-$k$ tokens with the highest importance scores in text $X$ and construct the important-token set $\mathbf{I}$:
\begin{equation}
\label{eq:3}
\mathbf{I}=\text{top-$k$}\left([ (x_t, s_t) \mid t = 1, 2, \dots, T ]\right).
\end{equation}

To mitigate the impact of discrepancies between the $\mathcal{M}_{sur}(.)$ and $\mathcal{M}_{tar}(.)$, we leverage the predictions of the target detector $\mathcal{M}_{tar}(.)$ to guide $\mathcal{F}_{\theta}(.)$ in more accurately identifying important tokens during adversarial training process. 
We detail the adversarial training process in \secref{Adversarial Training}.

\fakeparagraph{Embedding-level Perturbation.}
We apply a targeted perturbation on the embedding of the tokens in $\mathbf{I}$ to improve the attack effectiveness while preserving semantic integrity.
Formally, we introduce perturbations to the tokens in set $\mathbf{I}$ within the embeddings $E = [\mathbf{e}_1, \mathbf{e}_2, \ldots, \mathbf{e}_T]$ of the surrogate model to obtain the perturbed embedding $\tilde{E} = [\tilde{\mathbf{e}}_1, \tilde{\mathbf{e}}_2, \ldots, \tilde{\mathbf{e}}_T]$:
\begin{equation}
\label{eq:4}
\tilde{\mathbf{e}}_t = \mathbf{e}_t + \mathbf{1} _{[t \in \mathbf{I}]}\delta_t, 
\end{equation}
where $\mathbf{\delta}_t$ represents the perturbation of the $t$-th token, and $\mathbf{1} _{[t \in \mathbf{I}]}$ represents an indicator function with a value of $1$ if and only if the condition $t \in \mathbf{I}$ is satisfied, otherwise, it is $0$.
For the calculation of $\mathbf{\delta}_t$, we first initialize the perturbation from a normalized uniform distribution, and then design a single-step gradient ascent strategy to optimize the perturbation.
Specifically, we update the perturbation towards the direction where the KL divergence between the output distributions with respect to the original embedding $E$ and initial perturbed embedding $\tilde{E}^{0}$ increases most steeply.
This process is formulated as:
\begin{equation}
\footnotesize
\label{pre}
\vcenter{\hbox{$\displaystyle 
\begin{gathered}
\mathbf{\delta}_t^{0} \sim \mathcal{U}(a, b),\quad 
\hat{\mathbf{\delta}}_t^{0} = \xi \frac{\mathbf{\delta}_t^{0}}{\|\mathbf{\delta}_t^{0}\|_2}, \\[6pt]
\mathbf{\delta}_t = \epsilon\,  
\frac{
  \nabla_{\hat{\mathbf{\delta}}_t^{0}}\,\text{KL}\Bigl(P_{\text{sur}}(\mathbf{y}\mid E) \,\|\, P_{\text{sur}}(\mathbf{y}\mid \tilde{E}^{0})\Bigr)
}{
  \Bigl\|\nabla_{\hat{\mathbf{\delta}}_t^{0}}\,\text{KL}\Bigl(P_{\text{sur}}(\mathbf{y}\mid E) \,\|\, P_{\text{sur}}(\mathbf{y}\mid \tilde{E}^{0})\Bigr)\Bigr\|_2
},
\end{gathered}
$}}
\end{equation}
where $\hat{\mathbf{\delta}}_t^{0}$ represents the normalized value of $\mathbf{\delta}_{t}^{0}$, and $\xi$ is a scaling factor. The parameters $a$ and $b$ define the lower and upper bounds of the uniform distribution $\mathcal{U}(a, b)$, $\epsilon$ is a scaling factor, $P_{\text{sur}}(\mathbf{y} | E)$ and $P_{\text{sur}}(\mathbf{y} | \tilde{E}^{0})$ are label distribution before and after perturbation, respectively.

Afterwards, we project the $\tilde{E}$ back to the vocabulary with the language modeling head and select the top-$m$ tokens with the highest probabilities as candidates for replacing the important tokens:
\begin{equation}
\label{eq:7}
\footnotesize
\mathbf{C}_t = \text{top-$m$}(\mathrm{Softmax}(\mathrm{LMHead}(\mathcal{M}_{sur}(\tilde{E}))_t)), 
\end{equation}
where $\mathrm{LMHead}(\mathcal{M}_{sur}(\tilde{E}))_t$ represents the output of the $\mathrm{LMHead}(.)$ at position $t$.
Following POS Constraints \cite{zhou2024humanizing}, which require that the candidate words must match the part of speech of the words they replace, we filter the candidate tokens so that the candidate set contains the tokens with the same POS as the original one.

\subsubsection{Replace \& Refine}\label{AEE}
Based on the token importance calculated in \secref{AEG}, we introduce a greedy search and a greedy pruning strategy to efficiently construct powerful adversarial examples facilitating adversarial training.
 
\begin{algorithm}[t]
\small
\caption{Greedy Search Procedure}
\label{alg:algorithm2_short}
\begin{algorithmic}[1]
	\State {\textbf{Input:} Target Detector $\mathcal{M}_{tar}$, Original MGT $X$, Label $c$, Important-token Set $\mathbf{I}$.}
    \State $round \leftarrow 0$ and $\tilde{X} \leftarrow X$.
    \While{$round< round_{max}$}
       \State Get the token $x_{t}=\mathbf{I}[round]$ to be perturbed.
       \State{Compute corresponding perturbation $\mathbf{\delta}_t$ using Eq.\eqref{pre}.}
       \State{Get candidates $\mathbf{C}_t$ for replacing token using Eq.\eqref{eq:7}.}
       \State{Replace $x_t$ with the token in $\mathbf{C}_t$ to update $\tilde{X}$.}
       \State Classify $\tilde{X}$ via $\mathcal{M}_{tar}$ and obtain the output $\tilde{c}$.
       \If{$\tilde{c} \neq c$}
          \State \textbf{break} \quad // Attack success
       \Else
          \State $round \leftarrow round + 1$
          \quad // Attack failed
       \EndIf
    \EndWhile
    \State \textbf{Output:} Adversarial Example $\tilde{X}$.
\end{algorithmic}
\end{algorithm}

\begin{algorithm}[t]
\small
\caption{Greedy Pruning Procedure}
\label{alg:algorithm3_short}
\begin{algorithmic}[1]
	\State {\textbf{Input:} Adversarial Example $\tilde{X}$ from Algorithm \ref{alg:algorithm2_short}, Label $c$, Important-token Set $\mathbf{I}$ and Target Detector $\mathcal{M}_{tar}$.}
    \For{each perturbed token $\tilde{x}_t$ in $\tilde{X}$}
        \State {Revert $\tilde{x}_t$ to corresponding token $x_t$ in $X$.}
        \State Classify $\tilde{X}$ via $\mathcal{M}_{tar}$ and obtain the output $\tilde{c}$. 
        \If{$\tilde{c} \neq c$.}
            \State \textbf{continue} \quad // revert successful
        \Else
            \State Replace $x_t$ with $\tilde{x}_t$ again. \quad // revert failed
        \EndIf
    \EndFor
    \State \textbf{Output:} Final Adversarial Example $\tilde{X}$.
\end{algorithmic}
\end{algorithm}
\fakeparagraph{Greedy Search.} We present the process of greedy search in Algorithm \ref{alg:algorithm2_short}. 
For a piece of original MGT $X$, we substitute the token $x_t$ with the most possible candidate token in $\mathbf{C}_t$ sequentially according to the descending order of importance scores. 
After each replacement, we query the target model $\mathcal{M}_{tar}$ if the adversarial example in the current step is machine-generated.
We iterate this token-replacing procedure until the adversarial example deceives the $\mathcal{M}_{tar}$ or all the tokens in $\mathbf{I}$ are replaced. 
We then use the successful adversarial example (if the attack succeeds) or the text that is perturbed the most as training data for $\mathcal{M}_{tar}$.
Compared to existing methods \cite{liu2024hqa,yu2024query,hu2024fasttextdodger}, our method perturbs only the tokens in set $\mathbf{I}$ incrementally with the guidance of importance score, enabling the generation of adversarial examples with fewer queries and lower perturbation rates.

\fakeparagraph{Greedy Pruning.}
Due to the local optimality characteristic of greedy search~\cite{yu2024query, prim1957shortest}, the adversarial examples constructed by greedy search contain redundant perturbations. 
We apply the greedy pruning algorithm, as shown in Algorithm \ref{alg:algorithm3_short}, to further reduce the perturbation rate and make the attack stealthy without sacrificing its effectiveness.
Given an adversarial example $\tilde{X}$ generated with greedy search, we sample the perturbed token by order of importance scores and replace the selected token with its corresponding original token one-by-one.
After each restoration, we query the target model $\mathcal{M}_{tar}$ if $\tilde{X}$ is still a successful adversarial example.
If the attack remains successful after restoration,
the token is converted to the original one; otherwise,
we preserve the perturbed token.
We loop over all perturbed tokens and produce the final adversarial example
$\tilde{X}$.

To further validate the efficiency of our method, we provide a detailed theoretical analysis of query complexity and perturbation rate in \attackname in Appendix \ref{sec:analysis}. 

\subsection{\defensename for Defending Attacks}
\label{DA}
Our \defensename contains a target detector $\mathcal{M}_{tar}(.)$.
For the target detector $\mathcal{M}_{tar}(.)$, we expect it to learn to defend against adversarial attacks from the adversary and generalize the defense ability to other attacks. Specifically, given the target detector $\mathcal{M}_{tar}(.)$, for each original MGT $X_{i} \in \mathbf{X}$ and the corresponding adversarial example $\tilde{X}_{i}$ generated by the adversary, which share the same label $c_{i}$, the optimization objective of the target detector is to make correct predictions for both the original MGT and the adversarial example:
\begin{equation}
\label{eq:8}
\footnotesize
    \min_{\theta}(\sum_{X_{i}\in \mathbf{X}} (\mathcal{L}(\mathcal{M}^\theta_{tar}(X_{i}), c_{i})+\mathcal{L}(\mathcal{M}^\theta_{tar}(\tilde{X}_{i}), c_{i}))),
\end{equation}
where $\theta$ represents the learnable parameters of the target detector $\mathcal{M}_{tar}(.)$ and $\mathcal{L}$ is the loss function. 
This optimization objective aims to guide the $\mathcal{M}_{tar}(.)$ to simultaneously enhance its performance on both original MGT and adversarial examples, thereby compelling it to learn robust features of samples before and after attacks.
Regardless of whether the samples are subjected to adversarial interference, these features ensure that the detector can accurately classify the samples. As a result, the detector is better equipped to handle various types of attacks.

\subsection{Adversarial Training}
\label{AT}
We propose to train the adversary and detector in a co-training manner, that is, the two main components in \modelname are updated in the same training step. 
Unlike the previous methods \cite{zhang2024random, huang2024cert,zeng2023certified} that rely on static training set augmented with adversarial examples, synchronously updating the adversary and the detector allows the detector to learn from easy adversarial examples to hard ones, facilitating the defense to generalize to different attacks.

\noindent\textbf{Adversary Loss.}
The goal of the adversary is to precisely estimate the importance of each token in the detection to guide it to undertake a successful attack.
Thus, we use the gradient with respect to the input in calculating the cross-entropy loss on training data as the golden label for the token importance score.
However, since we are constrained in a black-box setting, we utilize the $\mathcal{M}_{sur}(.)$ to obtain the gradient.
This process is formulated as:
\begin{equation}
\label{eq:combined}  
\vcenter{\hbox{$\displaystyle 
\begin{gathered}
s^*_x = \Bigl\|\nabla_{x} \mathcal{L}_{\text{sur}}\Bigr\|_2, \\[6pt]
\mathcal{L}_{\text{sur}} = \frac{1}{M} \sum_{i=1}^{M} \left[ -\log P_{\text{sur}}(c_{i} \mid X_{i})\right],
\end{gathered}
$}}
\end{equation}
where $\mathcal{L}_{\text{sur}}$ is the cross-entropy classification loss of the surrogate model to the original MGTs.
Subsequently, we update the scoring network $\mathcal{F}_{\theta}(.)$ with the mean squared error loss:
\begin{equation}
\label{eq:12}
    \mathcal{L}_{\text{imp}} = \frac{1}{M} \sum_{i=1}^{M} \frac{1}{|X_i|} \sum_{x \in X_i} (s_x-s_x^*)^{2},
\end{equation}
where $|X_i|$ is the number of tokens in sample $X_i$. 
In addition, we leverage the output information from the target detector to guide the training of the $\mathcal{F}_{\theta}(.)$, thereby mitigating the impact of discrepancies between the surrogate and target detector.
Specifically, we replace the true labels in the cross-entropy loss with misleading labels to direct the $\mathcal{F}_{\theta}(.)$ to produce samples that are capable to deceive the detector:
\begin{equation}
\label{eq:13}
    \mathcal{L}_{\text{adv}} = \frac{1}{M} \sum_{i=1}^{M} \left[ -\text{log}P_{\text{tar}}(1-c_{i}|\tilde{X}_{i})\right].
\end{equation}
Finally, we balance the influence of various losses on the adversary by adjusting the weight parameter $\lambda$. The total loss for the adversary is given as follows:
\begin{equation}
\label{eq:14}
\mathcal{L}_{\text{A}} = \lambda \, \mathcal{L}_{\text{adv}} + (1-\lambda) \, \mathcal{L}_{\text{imp}}.
\end{equation}

\noindent\textbf{Detector Loss.}
The detector's goal is to maintain correct predictions in all circumstances.
Based on this, we optimize the detector towards minimizing a cross-entropy loss:
\begin{equation}
\label{eq:9}
\footnotesize
    \mathcal{L}_{\text{D}} = \frac{1}{M} \sum_{i=1}^{M} \left[ -\text{log}P_{\text{tar}}(c_{i}|X_{i})-\text{log}P_{\text{tar}}(c_{i}|\tilde{X}_{i})\right],
\end{equation}
where $P_{\text{tar}}(c_{i}|X_{i})$ is the probability of the target detector output at the correct label $c_{i}$. 

\noindent\textbf{Training Process.}
We update the adversary and the detector alternatively in the same training step.
Specifically, at each training step $t$, we first update the adversary with loss function \eqref{eq:13} while keeping the detector frozen.
The updated adversary generates adversarial example $\tilde{X}_i$ in the current step.
Then we update the detector with loss function \eqref{eq:9}.
We detail the adversarial training process in form of pseudocode in Appendix \ref{appd: pseudo}.

\label{Adversarial Training}

\section{Experiment Results}  
\label{sec:exp}
\vspace{-0.2cm}
We conduct extensive experiments to comprehensively evaluate the defense performance of our \defensename and also reveal the vulnerability of the current defense strategy with the adversarial examples generated by \attackname.

\begin{table*}[ht]
\centering
\renewcommand\arraystretch{1.8}
\footnotesize
\resizebox{\textwidth}{!}{
\begin{tabular}{c c c c c c c c c c c c c c c}
\toprule
\multirow{2}{*}{\textbf{Category}} 
& \multirow{2}{*}{\textbf{Method}} 
& \multirow{2}{*}{\textbf{Metric}} 
& \multicolumn{1}{c}{\textbf{Baseline}} 
& \multicolumn{6}{c}{\textbf{Adversarial Data Augmentation}} 
& \multicolumn{5}{c}{\textbf{Adversarial Training}} \\
\cmidrule(lr){4-4} \cmidrule(lr){5-10} \cmidrule(lr){11-15}
& & 
& F.t.XLM-RoBERTa-Base 
& EP & PP & CERT-ED & RanMask & Text-RS & Text-CRS 
& VAT & TAVAT & RADAR & OUTFOX & GREATER-D \\
\midrule

\multirow{12}{*}{\makecell{\textbf{Text} \\ \textbf{Perturbation}}}

& \textbf{\emph{Mixed Edit}} 
& \emph{ASR(\%)$\downarrow$}  
& 34.65 & \cellcolor{green!50}{4.58} & 32.33 & 16.74 & 17.00 & 22.81 & 15.34 & 33.19 & 37.50 & 18.76 & 11.33 & 8.85 \\

& \textbf{\emph{HMGC}} 
& \emph{ASR(\%)$\downarrow$}  
& 30.00 & 6.53 & 27.34 & 12.52 & 11.74 & 16.58 & 13.37 & 25.19 & 31.50 & 8.53 & 3.74 & \cellcolor{green!50}{2.28} \\

& \textbf{\emph{Paraphrasing}}
& \emph{ASR(\%)$\downarrow$} 
& 70.58 & 54.65 & 6.27 & 32.10 & 40.20 & 59.58 & 26.67 & 67.98 & 65.90 & \cellcolor{green!50}2.08 & 14.96 &  {3.45} \\

& \textbf{\emph{Code-switching MF}$^{*}$}
& \emph{ASR(\%)$\downarrow$} 
& 50.58 & 38.37 & 34.08 & 29.19 & 27.78 & 48.80 & 27.15 & 36.71 & 47.52 & 32.57 & 5.46 & \cellcolor{green!50}{1.13} \\

& \textbf{\emph{Code-switching MR}$^{*}$}
& \emph{ASR(\%)$\downarrow$}  
& 47.91 & 31.23 & 30.21 & 5.30  & 10.80 & 16.98 & 14.36 & 38.03 & 46.46 & 14.14 & 5.37 & \cellcolor{green!50}{1.02} \\

& \textbf{\emph{Human Obfuscation}}
& \emph{ASR(\%)$\downarrow$}  
& 18.42 & 22.19 & 27.00 & 13.64 & 18.86 & 18.05 & 15.24 & 25.35 & 24.17 & 16.26 & 5.74 & \cellcolor{green!50}{0.86} \\

& \textbf{\emph{Emoji-cogen}}
& \emph{ASR(\%)$\downarrow$} 
& 32.19 & 46.55 & 44.17 & 11.41 & 22.23 & 17.72 & 27.10 & 52.35 & 40.33 & 33.13 & 7.72 & \cellcolor{green!50}{0.47} \\

& \textbf{\emph{Typo-cogen}}
& \emph{ASR(\%)$\downarrow$} 
& 60.10 & 61.57 & 59.72 & 27.29 & 38.82 & 37.09 & 44.79 & 70.26 & 63.04 & 41.52 & 9.29 & \cellcolor{green!50}{1.08} \\

& \textbf{\emph{ICL}} 
& \emph{ASR(\%)$\downarrow$}
& 1.40 & 1.41 & 1.30 & 1.72 & 0.83 & 1.89 & 0.67 & 1.13 & 1.88 & 2.77 & 0.59 &\cellcolor{green!50}{0.20} \\

& \textbf{\emph{Prompt Paraphrasing}}
& \emph{ASR(\%)$\downarrow$} 
& \cellcolor{green!50}{0.00} & 0.69 & \cellcolor{green!50}{0.00} & 0.70 & \cellcolor{green!50}0.00 & \cellcolor{green!50}{0.00} & \cellcolor{green!50}0.00 & \cellcolor{green!50}0.00 &  \cellcolor{green!50}0.00 & 0.94 & 0.65 & 0.23 \\

& \textbf{\emph{CSGen}}
& \emph{ASR(\%)$\downarrow$}  
& 25.44 & 22.81 & 6.14 & 10.65 & 1.70 & 2.41 & 23.95 & 26.88 & 27.52 & 17.07 & 3.40 & \cellcolor{green!50}0.00 \\

\cdashline{2-15}[1pt/1pt]

& \textbf{\emph{Avg.}} 
& \emph{ASR(\%)$\downarrow$}  
& 33.75	& 26.42	& 24.41	& 14.66	& 17.27	& 21.99	& 18.97	& 34.28	& 35.07	& 17.07	& 6.20	& \cellcolor{green!50}1.78
\\

\midrule

\multirow{14}{*}{\makecell{\textbf{Adversarial} \\ \textbf{Attack}} }

& \multirow{2}{*}{\textbf{\emph{PWWS}}} 
& \emph{ASR(\%)$\downarrow$}
& 61.45 & 48.47 & 49.06 & 11.07 & 14.83 & 16.13 & 19.15 & 53.56 & 62.68 & 12.58 & - & \cellcolor{green!50}5.91 \\
& 
& \emph{Queries$\uparrow$} 
& 1197.95 & 1237.62 & 1235.01 & 1371.28 & 1361.15 & 1356.67 & 1342.15 & 1226.53 & 1180.87 & 1366.85 & - &\cellcolor{green!50}1390.69 \\

& \multirow{2}{*}{\textbf{\emph{TextFooler}}} 
& \emph{ASR(\%)$\downarrow$} 
& 72.29 & 70.76 & 70.02 & 11.07 & 17.43 & 19.56 & 22.98 & 69.25 & 94.02 & 14.46 & - & \cellcolor{green!50}6.11 \\
& 
& \emph{Queries$\uparrow$} 
& 690.92 & 724.23 & 713.77 & 1362.17 & 1291.35 & 1271.34 & 1237.18 & 780.51 & 440.28 & 1301.84 & - & \cellcolor{green!50}1396.52 \\

& \multirow{2}{*}{\textbf{\emph{BERTAttack}}} 
& \emph{ASR(\%)$\downarrow$}
& 71.49 & 69.34 & 63.52 & 6.04 & 12.42 & 12.30 & 16.94 & 69.45 & 93.81 & 13.08 & - & \cellcolor{green!50}5.70 \\
& 
& \emph{Queries$\uparrow$} 
& 411.54 & 446.64 & 425.68 & 710.56 & 682.29 & 680.14 & 667.82 & 450.57 & 279.02 & 677.24 & - & \cellcolor{green!50}718.84 \\

& \multirow{2}{*}{\textbf{\emph{A2T}}} 
& \emph{ASR(\%)$\downarrow$}
& 45.47 & 37.45 & 50.10 & 6.24 & 11.22 & 14.52 & 14.31 & 29.12 & 54.64 & 8.58 & - & \cellcolor{green!50}5.09 \\
& 
& \emph{Queries$\uparrow$} 
& 293.26 & 320.88 & 302.18 & 516.07 & 499.89 & 493.25 & 488.91 & 361.88 & 207.15 & 505.54 & - & \cellcolor{green!50}519.77 \\

& \multirow{2}{*}{\textbf{\emph{T-PGD}}} 
& \emph{ASR(\%)$\downarrow$}
& 49.90	& 34.15	& 25.28	& 8.52	& 12.07	& 13.15	& 14.92	& 42.14	& 45.58	& 7.79	& -	& \cellcolor{green!50}5.61 \\
& 
& \emph{Queries$\uparrow$} 
& 354.75 & 515.62 & 712.08 & 799.62 & 766.53 & 752.04	& 744.17 & 404.85 & 377.69 & 813.76	& -	& \cellcolor{green!50}824.58 \\

& \multirow{2}{*}{\textbf{\emph{GREATER-A(ours)}}} 
& \emph{ASR(\%)$\downarrow$} 
& 96.58 & 87.08 & 84.34 & 62.17 & 63.58 & 75.25 & 82.29 & 89.26 & 85.02 & 71.04 & - & \cellcolor{green!50}46.08 \\
& 
& \emph{Queries$\uparrow$} 
& 62.63 & 66.71 & 68.53 & 99.56 & 98.92 & 106.08 & 75.98 & 63.57 & 67.17 & 100.74 & - & \cellcolor{green!50}190.64 \\

\cdashline{2-15}[1pt/1pt]

& \multirow{2}{*}{\textbf{\emph{Avg.}}} 
& \emph{ASR(\%)$\downarrow$} 
& 66.20	& 57.87	& 57.05	& 17.52	& 21.92	& 25.15	& 28.43	& 58.80	& 72.62	& 21.26	& - & \cellcolor{green!50}12.42 \\
& 
& \emph{Queries$\uparrow$}
& 501.84 & 551.95 & 576.21	& 809.88 & 783.36	& 776.59	& 759.37	& 547.98	& 425.36	& 794.33	& -	& \cellcolor{green!50}840.17
\\

\midrule

\multirow{1}{*}{\textbf{Total}} & \textbf{\emph{Avg.}} & \emph{ASR(\%)$\downarrow$} 
& 45.20	& 37.52	& 35.93	& 15.67	& 18.91	& 23.11	& 22.31	& 42.93	& 48.33	& 18.55	& 6.20	& \cellcolor{green!50}5.53\\
\bottomrule
\end{tabular}
}
\caption{\textbf{Performance of defense methods under different attacks.} 
The best results are highlighted in \colorbox{green!50}{green} background.
${*}$ means that Code-switching MF and Code-switching MR are two variations of Code-switching method. Since the detector in OUTFOX is a closed-source LLM, it cannot be attacked by adversarial attack. Therefore, we only evaluate OUTFOX under text perturbation attacks.}
\label{tab:all_results_modified}
\end{table*}

\subsection{Defense Performance for \defensename}
\fakeparagraph{Experiment Setting.}
We evaluate our defense model \defensename against 16 text perturbation and adversarial attack methods, whose detailed introductions are outlined in the Appendix~\ref{sec:pert}. 
The competitors include \textbf{i) data augmentation methods:} Editing Pretrained (EP) \cite{wang2024comprehensive}, Paraphrasing Pretrained (PP) \cite{wang2024comprehensive}, CERT-ED \cite{huang2024cert}, RanMask \cite{zeng2023certified}, Text-RS \cite{zhang2024random}, Text-CRS \cite{zhang2024text}.
\textbf{ii) adversarial training methods:}
Virtual Adversarial Training (VAT) \cite{miyato2016virtual}, Token Aware Virtual Adversarial Training (TAVAT) \cite{li2021tavat}, RADAR \cite{hu2023radar}, OUTFOX \cite{koike2024outfox}.
Detailed introduction and implementation are presented in Appendix \ref{apdx:defmethod} and \ref{appdx:implementation}. The dataset we use is presented in Appendix~\ref{apdx:dataset}.

\fakeparagraph{Experiment results.}
We present the defense performance of \defensename in Table~\ref{tab:all_results_modified} and unveil the following three key insights:
1) \textbf{Best defense performance.}
Our method exhibits the best defense performance against different attacks among all competitors.
The average Attack Success Rate (ASR) drops to \textbf{1.78\%} for text perturbation attacks and \textbf{12.42\%} for adversarial attacks, which are lower than the second-best methods, \textbf{RADAR (6.20\%)} and \textbf{CERT-ED (17.52\%)}, by \textbf{4.42\%} and \textbf{5.10\%}, respectively.
2) \textbf{Most effort needed for adversarial attack.}
We observe that it takes adversarial attack more resources to conduct a successful attack to \defensename.
\textbf{840.17} queries are required on average, which is \textbf{30.29} more than other defense methods.
Moreover, the average ASR against \defensename is only \textbf{12.42\%} for adversarial attacks. 
It illustrates that \defensename makes adversarial attacks both inefficient and ineffective.
3) \textbf{Generalized Defense to Different Attacks.}
We notice that \defensename significantly reduces the ASR of different kinds of attacks even though it is trained with \attackname.
As an exemplary method, EP performs better than \defensename when defending Mixed Edit Attack but cannot withstand other attacks.
The defense against a wide variety of attacks demonstrates the generalized defense effect of \defensename.
\vspace{-0.1cm}

\subsection{Attack Performance for \attackname}
\label{sec:GeneratorResult}
In this section, we evaluate the effectiveness of \attackname in the black-box setting using the metrics detailed in Appendix~\ref{apdx:metric}.
We categorize the comparison methods into two classes:
i) \textit{Query-based methods}, which query the target model for output to adjust attack strategy, including PWWS~\cite{ren2019generating}, TextFooler~\cite{jin2020bert}, BERTAttack~\cite{li2020bert},
HQA~\cite{liu2024hqa}, ABP~\cite{yu2024query}, T-PGD~\cite{yuan2023bridge}, and FastTextDodger~\cite{hu2024fasttextdodger}.
ii) \textit{Zero-query methods}, which conducts attack without any information from the target model, including WordNet~\cite{zhou2024humanizing}, Back Translation~\cite{zhou2024humanizing}, Rewrite~\cite{zhou2024humanizing}, 
T-PGD~\cite{yuan2023bridge}, and HMGC~\cite{zhou2024humanizing}.
To further demonstrate the effectiveness of \attackname, we also incorporate A2T~\cite{yoo2021towards},  a SOTA white-box method in query-based methods.
Detailed introduction of these attack methods are listed in Appendix~\ref{sec:attack}.
Note that \attackname is a query-based method.
However, for a fair comparison, we also implement our method in a zero-query setting where we query the surrogate model for feedback. 
Among all the attacks, we employ a fine-tuned XLM-RoBERTa-Base model \cite{conneau2019unsupervised}
as the target detector.

\fakeparagraph{Experiment Results.}
We show the experiment results in Table \ref{tab:merged_attack_results}.
We find \attackname performs the best in three dimensions: effectiveness, efficiency, and stealthy.
\attackname achieves 96.58\% and 69.11\% in terms of ASR in query and zero-query settings, respectively, which significantly outperforms all other methods.
Moreover, in the query-based setting, it only takes 62.63 queries for \attackname to conduct a successful attack, which is four times fewer than its competitors.
As for the stealthy, the texts edited by \attackname achieves the lowest perplexity and has the best readability implied by the highest USE and lowest readability change in query-based scenario.
In the zero-query setting, \attackname performs second-best in terms of readability after Back Translation which can rarely conduct a successful attack.

\begin{table}[ht]
\centering
\renewcommand{\arraystretch}{1.8}
\resizebox{0.45\textwidth}{!}{%
\begin{tabular}{c c c c c c c c}
\toprule
\textbf{{Attack Type}} & \textbf{Method} & \textbf{Avg Queries $\downarrow$} & \textbf{ASR (\%) $\uparrow$} & \textbf{Pert. (\%) $\downarrow$} & \(\Delta\mathrm{PPL}\) $\downarrow$ & \textbf{USE $\uparrow$} & \(\Delta r\) $\downarrow$ \\
\midrule
\multirow{9}{*}{\textbf{Query-based}} 
    & PWWS                & 1197.95               & 61.45                      & \cellcolor{green!50}4.71 & 37.85   & 0.9488 & 12.76  \\
    & TextFooler          & 690.92                & 72.29                      & 6.26                          & 46.89   & 0.9302 & 21.07  \\
    & BERTAttack          & 411.54                & 71.49                      & 5.79                         & 36.15   & 0.9402 & 15.78  \\
    & HQA                 & 283.89                & 88.13                      & 23.57                         & 102.87  & 0.8854 & 72.16  \\
    & FastTextDodger      & 745.75                & 63.78                      & 13.29                         & 76.15   & 0.9188 & 55.14  \\
    & ABP                 & 785.18                & 75.65                      & 14.63                         & 39.61   & 0.8709 & 26.40  \\
    & T-PGD               & 354.75                & 49.90                      & 38.01                         & 181.31  & 0.8197 & 35.09  \\
    & \attackname          & \cellcolor{green!50}62.63 & \cellcolor{green!50}96.58 & 7.26              & \cellcolor{green!50}35.22  & \cellcolor{green!50}0.9506 & \cellcolor{green!50}9.21 \\
    & A2T (White Box)     & 293.26                & 45.47                      & 7.01                          & 62.84   & 0.9215 & 32.07  \\
\midrule
\multirow{6}{*}{\textbf{Zero-query}}
    & WordNet             & -                     & 42.60                      & -                               & 26.27   & 0.90   & 4.26   \\
    & Back Translation    & -                     & 2.40                       & -                               & \cellcolor{green!50}6.40   & 0.91   & \cellcolor{green!50}3.06   \\
    & Rewrite             & -                     & 36.47                      & -                               & 92.08   & 0.79   & 15.62  \\
    & T-PGD               & -                     & 62.40                      & -                               & 140.13  & 0.82   & 44.99  \\
    & HMGC                & -                     & 30.00                      & -                               & 12.94   & 0.84   & 36.90  \\
    & \attackname          & -                     & \cellcolor{green!50}69.11                      & -                               & 43.11   & \cellcolor{green!50}0.92  & 4.55   \\
\bottomrule
\end{tabular}%
}
\caption{\textbf{Attack results of the query and zero-query attack methods on the target model.} 
Note that perturbation rates are not reported for zero-query methods because the zero-query methods rewrite the whole text.
The best result in each group is highlighted with a \colorbox{green!50}{green} background.}
\label{tab:merged_attack_results}
\end{table}
\vspace{-0.1cm}

\section{Discussion}
\label{sec:discussion}

\subsection{Impact of Attack Strength of \attackname on \defensename}
In this section, we investigate how the strength of adversarial training affects the performance of both our \defensename and \attackname.
We define the attack strength as the max number of query the adversary is allowed to make in the training process.
Generally, if the adversary queries the target model more frequently, its final output tends to be more effective.
\begin{figure}[t]
    \centering
    \resizebox{0.48\textwidth}{!}{    \includegraphics{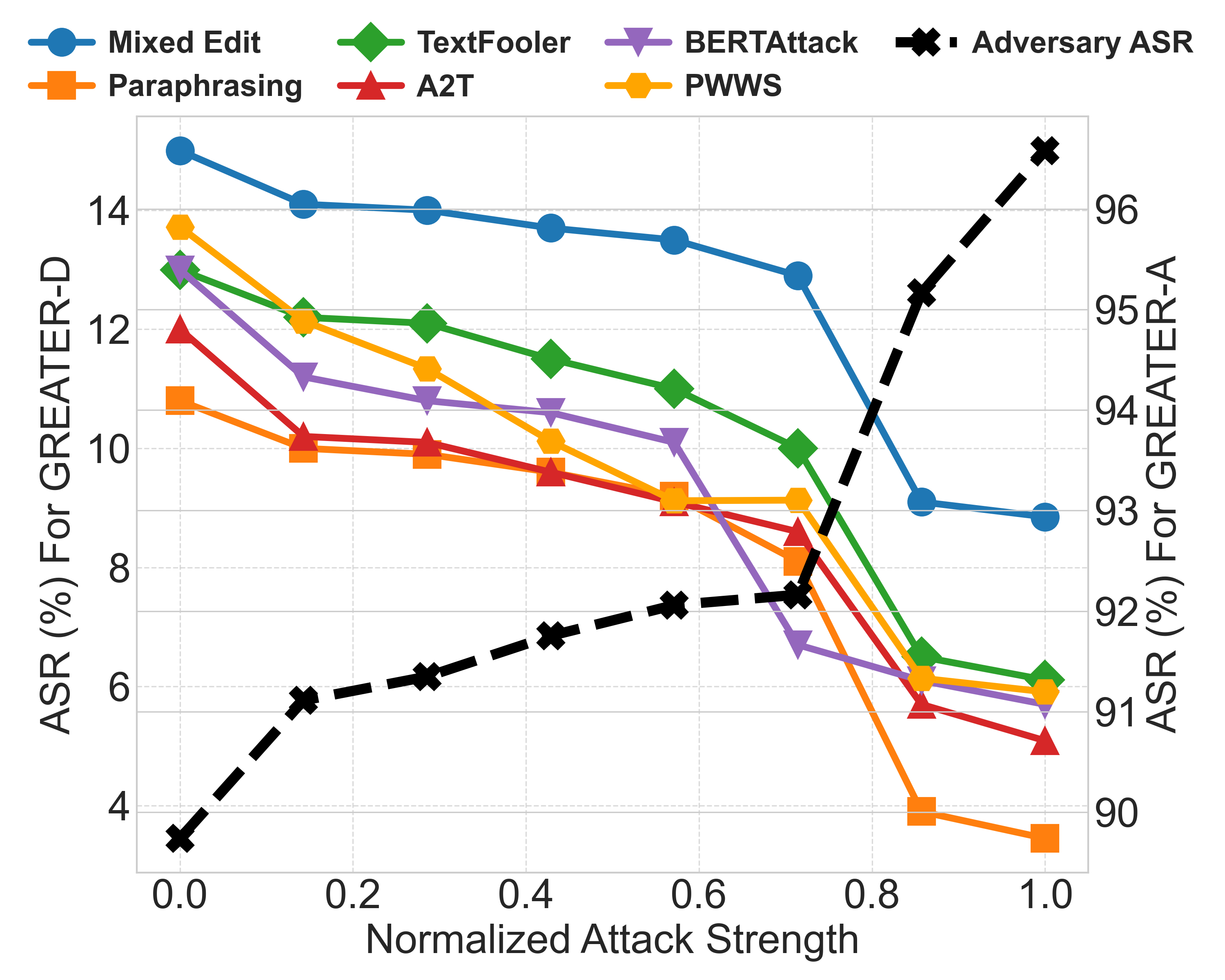} 
    }
	\caption{
    \textbf{Impact of attack strength in \modelname.}
    We normalize the attack strength for better visualization of the results.
	}
    \label{fig:compare} 
    \vspace{-0.5cm}
\end{figure}
We increase attack strength in \modelname and evaluate the performance of \defensename under 7 attacks and \attackname on the fine-tuned \texttt{XLM-RoBERTa-Base} detector and present the results in Figure~\ref{fig:compare}. 
The ASR significantly increases as the attack strength grows, which proves the rationale of the attack strength measure.
We observe that the \defensename becomes more robust with the increasing of attack strength, indicated by the decreasing ASR under all attacks.
Notably, the ASR under \textit{Paraphrasing} Attack decreases from \textbf{10.80} to \textbf{3.45} as the attack strength increases from \textbf{0.0} to \textbf{1.0}, which is \textbf{3.13} times lower.

The experimental results indicate that attack strength is a key factor influencing the robustness of the MGT detector. 
However, an increased number of queries comes with more cost on time and budget.
Thus, there exists a trade-off between the effectiveness and efficiency in adversarial training.

\subsection{Defense Adaptation to Different Backbones}
\label{sec:replacement}
To demonstrate the generalizability and effectiveness of our \modelname across different model architectures, we replace the backbone model with seven state-of-the-art transformer-based models, including both base and large variants of ALBERT~\cite{lan2019albert}, DeBERTa~\cite{he2020deberta}, RoBERTa~\cite{liu2019roberta}, and XLM-RoBERTa~\cite{conneau2019unsupervised}.
We report the ASR of each model under the Paraphrasing Attack with a budget of $0.74$, both before defense (Baseline) and after applying our \defensename{}. We compare it with the best-performing defense method CERT-ED and show the result in Table~\ref{tab:performance_comparison}.

\begin{table}[ht]
\centering
\renewcommand\arraystretch{1.4}
\footnotesize
\resizebox{0.45\textwidth}{!}{
\begin{tabular}{lcccc}
\toprule
\textbf{Model} & \textbf{Metric} & \textbf{Baseline} & \textbf{CERT-ED} & \textbf{\defensename} \\
\midrule

ALBERT Base (12M)     & \emph{ASR(\%)$\downarrow$} & 63.58   & 46.71   & \cellcolor{green!50}35.62 \\
ALBERT Large (18M)    & \emph{ASR(\%)$\downarrow$} & 97.14   & 45.45   & \cellcolor{green!50}43.40 \\
DeBERTa Base (86M)     & \emph{ASR(\%)$\downarrow$} & 22.95   & 29.37   & \cellcolor{green!50}4.10 \\
DeBERTa Large (304M)    & \emph{ASR(\%)$\downarrow$} & 51.72   & 25.07   & \cellcolor{green!50}10.25 \\
RoBERTa Base (125M)     & \emph{ASR(\%)$\downarrow$} & 81.88   & 26.78   & \cellcolor{green!50}15.32 \\
RoBERTa Large (355M)    & \emph{ASR(\%)$\downarrow$} & 30.21   & 34.66   & \cellcolor{green!50}5.38 \\
XLM-RoBERTa Large (561M) & \emph{ASR(\%)$\downarrow$} & 81.99   & 40.04   & \cellcolor{green!50}1.73 \\
\bottomrule
\end{tabular}
}
\caption{\textbf{Experiment results of various models before and after applying \defensename under Paraphrasing attack.} The best result in each group is highlighted with a \colorbox{green!50}{green} background.}
\label{tab:performance_comparison}
\end{table}

The results presented in Table~\ref{tab:performance_comparison}  demonstrate the efficacy of our adversarial training method across a diverse set of transformer-based models. Notably, \textit{XLM-RoBERTa-Large} exhibits the most substantial improvement, with ASR decreasing by 80.26\%, 38.31\% compared with baseline and CERT-ED, highlighting the significant impact of adversarial training on models with initially lower performance metrics. Similarly, both \textit{RoBERTa Base} and \textit{DeBERTa Large} demonstrate substantial improvements. Specifically, \textit{RoBERTa Base} achieves a 66.56\% reduction in ASR compared to the baseline and outperforms the CERT-ED with an additional 11.46\% decrease.
Likewise, \textit{DeBERTa Large} exhibits a 41.47\% drop in ASR relative to the baseline and surpasses CERT-ED by reducing ASR by an additional 14.82\%. 
These results underscore the robustness and versatility of our adversarial training approach across different model sizes and architectures. While smaller models like \textit{ALBERT Base} and \textit{ALBERT Large} exhibit more modest gains, the consistent upward trends across all evaluated models affirm that our adversarial training method effectively enhances model resilience and performance against Paraphrasing Attack. This versatility makes our approach a valuable tool for improving a wide range of transformer-based models in adversarial settings.

\subsection{Impact of Surrogate Model of \attackname}

We use different surrogate model and dataset for training GREATER-A to demonstrate that GREATER-A's performance is independent of surrogate model selection and training data. Specifically, we use RoBERTa-Large and GPT2 as surrogate models, and use SemEval (same dataset for training GREATER-D, abbreviated as IND) and M4~\citep{wang2024m4} (different dataset from training GREATER-D, abbreviated as IND) to validate our claim. The results are shown in Table \ref{tab:attack_results_surrogate}.

\begin{table}[ht]
\centering
\renewcommand\arraystretch{1.8}
\resizebox{0.45\textwidth}{!}{
\begin{tabular}{l l c c c c c c}
\toprule
\textbf{Surrogate Model} & \textbf{Dataset (Type)} & \textbf{Avg Queries} ↓ & \textbf{ASR (\%)} ↑ & \textbf{Pert. (\%)} ↓ & $\Delta$PPL ↓ & \textbf{USE} ↑ & $\Delta r$ ↓ \\
\midrule
RoBERTa-Large & SemEval (IND) & \cellcolor{green!50}62.63 & \cellcolor{green!50}96.58 & \cellcolor{green!50}7.26 & \cellcolor{green!50}35.22 & 0.9506 & \cellcolor{green!50}9.21 \\
RoBERTa-Large & M4 (OOD)      & 65.77 & 95.72 & 7.33 & 38.84 & 0.9503 & 9.65 \\
GPT2          & SemEval (IND) & 65.60 & 96.17 & 7.29 & 36.87 & \cellcolor{green!50}0.9550 & 12.73 \\
GPT2          & M4 (OOD)      & 63.80 & 96.36 & 7.27 & 37.14 & 0.9536 & 10.64 \\
\bottomrule
\end{tabular}
}
\caption{\textbf{The attack results of different surrogate models and datasets that the surrogate models are trained with.} The best result in each group is highlighted with a \colorbox{green!50}{green} background.}
\label{tab:attack_results_surrogate}
\end{table}

As shown in Table \ref{tab:attack_results_surrogate}, the attack effectiveness remains stable regardless of the surrogate model or training dataset. This consistency in attack success rates indicates that GREATER-A’s performance is not sensitive to changes in the surrogate model or training data, supporting the claim that the use of a surrogate model from the same family as the target model does not leak crucial information.

\section{Conclusion}
\vspace{-0.2cm}
In this paper, we proposed an adversarial training framework \textbf{GRE}edy \textbf{A}dversary
Promo\textbf{T}ed Defend\textbf{ER} (\modelname) to enhance the robustness of MGT detector under different text perturbation and adversarial attacks.
We design a novel attack strategy for the adversary including Identify \& Perturb and Replace \& Refine to construct effective adversarial examples efficiently.
In \modelname, we update the adversary and the detector alternatively in the same training step for better defense generalization.
Our experiment results demonstrate the efficacy of our detector \defensename under 16 attacks along with the leading performance of adversary \attackname compared with 13 methods.
The discussion on the relationship between attack strength and defense performance reveals the importance of a powerful adversary in adversarial training for robust MGT detectors.

\section*{Limitations}

Despite the promising results achieved by \modelname, there are three primary limitations to this study. 
\textbf{First}, our defense method can generalize to different attacks but its application on texts in different languages remains a challenge. 
\textbf{Second}, the computational cost of training the adversarial framework, particularly for the adversary and detector, is substantial, requiring significant hardware resources that could limit its deployment in resource-constrained settings.
\textbf{Third}, the detector presented in our work is only able to attribute the origin of the text on the document-level.
However, more works are needed for fine-grained (\eg sentence-level, token-level) detection in Human-AI co-authored texts.

\section*{Ethics Statement}

This study seeks to improve the robustness and security of machine-generated text (MGT) detection, with a focus on defending against adversarial threats. While the proposed \modelname\ framework enhances detection capabilities, it also introduces potential risks of misuse, such as creating adversarial examples to evade detection systems. To mitigate such risks, our experiments were conducted in controlled environments, and details that could enable misuse were abstracted. We emphasize that this work is intended solely for advancing detection technologies and defending against malicious applications. Ethical use of these findings is imperative, and any misuse for harmful purposes is strongly discouraged.
The artifacts used in our work are all under the restriction of the license.
\section*{Acknowledgment}
We thank all the reviewers and the area chair for their helpful feedback, which aided us in greatly improving the paper.
This work is supported by National Key R\&D Program (2023YFB3107400), National Natural Science Foundation of China (62272371, 62103323, U24B20185, T2442014, 62161160337, 62132011, U21B2018), Initiative Postdocs Supporting Program (BX20190275, BX20200270), China Postdoctoral Science Foundation (2019M663723, 2021M692565), Fundamental Research Funds for the Central Universities under grant (xzy012024144), and Shaanxi Province Key Industry Innovation Program (2023-ZDLGY-38, 2021ZDLGY01-02). Thanks to the New Cornerstone Science Foundation and the Xplorer Prize.

\bibliography{GENSHIN}     

\begin{thebibliography}{65}
\providecommand{\natexlab}[1]{#1}

\bibitem[{Achiam et~al.(2023)Achiam, Adler, Agarwal, Ahmad, Akkaya, Aleman, Almeida, Altenschmidt, Altman, Anadkat et~al.}]{achiam2023gpt}
Josh Achiam, Steven Adler, Sandhini Agarwal, Lama Ahmad, Ilge Akkaya, Florencia~Leoni Aleman, Diogo Almeida, Janko Altenschmidt, Sam Altman, Shyamal Anadkat, et~al. 2023.
\newblock Gpt-4 technical report.
\newblock \emph{arXiv preprint arXiv:2303.08774}.

\bibitem[{Anthropic(2024)}]{claude3}
Anthropic. 2024.
\newblock \href {https://www.anthropic.com} {Claude 3: A conversational ai model}.

\bibitem[{Bao et~al.()Bao, Zhao, Teng, Yang, and Zhang}]{baofast}
Guangsheng Bao, Yanbin Zhao, Zhiyang Teng, Linyi Yang, and Yue Zhang.
\newblock Fast-detectgpt: Efficient zero-shot detection of machine-generated text via conditional probability curvature.
\newblock In \emph{The Twelfth International Conference on Learning Representations}.

\bibitem[{Bao et~al.(2023)Bao, Zhao, Teng, Yang, and Zhang}]{bao2023fast}
Guangsheng Bao, Yanbin Zhao, Zhiyang Teng, Linyi Yang, and Yue Zhang. 2023.
\newblock Fast-detectgpt: Efficient zero-shot detection of machine-generated text via conditional probability curvature.
\newblock \emph{arXiv preprint arXiv:2310.05130}.

\bibitem[{Biggio and Roli(2018)}]{biggio2018wild}
Battista Biggio and Fabio Roli. 2018.
\newblock Wild patterns: Ten years after the rise of adversarial machine learning.
\newblock In \emph{Proceedings of the 2018 ACM SIGSAC Conference on Computer and Communications Security}, pages 2154--2156.

\bibitem[{Cer et~al.(2018)Cer, Yang, Kong, Hua, Limtiaco, John, Constant, Guajardo-Cespedes, Yuan, Tar et~al.}]{cer2018universal}
Daniel Cer, Yinfei Yang, Sheng-yi Kong, Nan Hua, Nicole Limtiaco, Rhomni~St John, Noah Constant, Mario Guajardo-Cespedes, Steve Yuan, Chris Tar, et~al. 2018.
\newblock Universal sentence encoder.
\newblock In \emph{Proceedings of the 2018 Conference on Empirical Methods in Natural Language Processing: System Demonstrations}, pages 169--174.

\bibitem[{Conneau et~al.(2019)Conneau, Khandelwal, Goyal, Chaudhary, Wenzek, Guzm{\'a}n, Grave, Ott, Zettlemoyer, and Stoyanov}]{conneau2019unsupervised}
Alexis Conneau, Kartikay Khandelwal, Naman Goyal, Vishrav Chaudhary, Guillaume Wenzek, Francisco Guzm{\'a}n, Edouard Grave, Myle Ott, Luke Zettlemoyer, and Veselin Stoyanov. 2019.
\newblock Unsupervised cross-lingual representation learning at scale.
\newblock \emph{arXiv preprint arXiv:1911.02116}.

\bibitem[{Dubey et~al.(2024)Dubey, Jauhri, Pandey, Kadian, Al-Dahle, Letman, Mathur, Schelten, Yang, Fan et~al.}]{dubey2024llama}
Abhimanyu Dubey, Abhinav Jauhri, Abhinav Pandey, Abhishek Kadian, Ahmad Al-Dahle, Aiesha Letman, Akhil Mathur, Alan Schelten, Amy Yang, Angela Fan, et~al. 2024.
\newblock The llama 3 herd of models.
\newblock \emph{arXiv preprint arXiv:2407.21783}.

\bibitem[{Eger et~al.(2019)Eger, {\c{S}}ahin, R{\"u}ckl{\'e}, Lee, Schulz, Mesgar, Swarnkar, Simpson, and Gurevych}]{eger2019text}
Steffen Eger, G{\"o}zde~G{\"u}l {\c{S}}ahin, Andreas R{\"u}ckl{\'e}, Ji-Ung Lee, Claudia Schulz, Mohsen Mesgar, Krishnkant Swarnkar, Edwin Simpson, and Iryna Gurevych. 2019.
\newblock Text processing like humans do: Visually attacking and shielding nlp systems.
\newblock In \emph{Proceedings of the 2019 Conference of the North American Chapter of the Association for Computational Linguistics: Human Language Technologies, Volume 1 (Long and Short Papers)}, pages 1634--1647.

\bibitem[{Fishchuk(2023)}]{fishchuk2023adversarial}
Vitalii Fishchuk. 2023.
\newblock Adversarial attacks on neural text detectors.
\newblock {B.S.} thesis, University of Twente.

\bibitem[{Flesch(1948)}]{flesch1948new}
Rudolf Flesch. 1948.
\newblock A new readability yardstick.
\newblock \emph{Journal of Applied Psychology}, 32(3):221--233.

\bibitem[{Gabrilovich and Gontmakher(2002)}]{gabrilovich2002homograph}
Evgeniy Gabrilovich and Alex Gontmakher. 2002.
\newblock The homograph attack.
\newblock \emph{Communications of the ACM}, 45(2):128.

\bibitem[{Guo et~al.(2025)Guo, Yang, Zhang, Song, Zhang, Xu, Zhu, Ma, Wang, Bi et~al.}]{guo2025deepseek}
Daya Guo, Dejian Yang, Haowei Zhang, Junxiao Song, Ruoyu Zhang, Runxin Xu, Qihao Zhu, Shirong Ma, Peiyi Wang, Xiao Bi, et~al. 2025.
\newblock Deepseek-r1: Incentivizing reasoning capability in llms via reinforcement learning.
\newblock \emph{arXiv preprint arXiv:2501.12948}.

\bibitem[{Guo et~al.()Guo, He, Zhang, Zhang, Feng, Huang, and Ma}]{guodetective}
Xun Guo, Yongxin He, Shan Zhang, Ting Zhang, Wanquan Feng, Haibin Huang, and Chongyang Ma.
\newblock Detective: Detecting ai-generated text via multi-level contrastive learning.
\newblock In \emph{The Thirty-eighth Annual Conference on Neural Information Processing Systems}.

\bibitem[{Hans et~al.(2024)Hans, Schwarzschild, Cherepanova, Kazemi, Saha, Goldblum, Geiping, and Goldstein}]{hans2024spotting}
Abhimanyu Hans, Avi Schwarzschild, Valeriia Cherepanova, Hamid Kazemi, Aniruddha Saha, Micah Goldblum, Jonas Geiping, and Tom Goldstein. 2024.
\newblock Spotting llms with binoculars: Zero-shot detection of machine-generated text.
\newblock \emph{arXiv preprint arXiv:2401.12070}.

\bibitem[{He et~al.(2020)He, Liu, Gao, and Chen}]{he2020deberta}
Pengcheng He, Xiaodong Liu, Jianfeng Gao, and Weizhu Chen. 2020.
\newblock Deberta: Decoding-enhanced bert with disentangled attention.
\newblock \emph{arXiv preprint arXiv:2006.03654}.

\bibitem[{Helsinki-NLP(2020)}]{helsinki-nlp}
Helsinki-NLP. 2020.
\newblock Helsinki-nlp machine translation models.
\newblock \url{https://huggingface.co/Helsinki-NLP}.

\bibitem[{Hu et~al.(2023{\natexlab{a}})Hu, Chen, and Ho}]{hu2023radar}
Xiaomeng Hu, Pin-Yu Chen, and Tsung-Yi Ho. 2023{\natexlab{a}}.
\newblock Radar: Robust ai-text detection via adversarial learning.
\newblock \emph{Advances in neural information processing systems}, 36:15077--15095.

\bibitem[{Hu et~al.(2024)Hu, Liu, Zheng, Zhao, Wang, Zhang, and Du}]{hu2024fasttextdodger}
Xiaoxue Hu, Geling Liu, Baolin Zheng, Lingchen Zhao, Qian Wang, Yufei Zhang, and Minxin Du. 2024.
\newblock Fasttextdodger: Decision-based adversarial attack against black-box nlp models with extremely high efficiency.
\newblock \emph{IEEE Transactions on Information Forensics and Security}.

\bibitem[{Hu et~al.(2023{\natexlab{b}})Hu, Chen, Wu, Wu, Zhang, and Huang}]{hu2023unbiased}
Zhengmian Hu, Lichang Chen, Xidong Wu, Yihan Wu, Hongyang Zhang, and Heng Huang. 2023{\natexlab{b}}.
\newblock Unbiased watermark for large language models.
\newblock \emph{arXiv preprint arXiv:2310.10669}.

\bibitem[{Huang et~al.(2024)Huang, Marchant, Ohrimenko, and Rubinstein}]{huang2024cert}
Zhuoqun Huang, Neil~G Marchant, Olga Ohrimenko, and Benjamin~IP Rubinstein. 2024.
\newblock Cert-ed: Certifiably robust text classification for edit distance.
\newblock \emph{arXiv preprint arXiv:2408.00728}.

\bibitem[{Jin et~al.(2020)Jin, Jin, Zhou, and Szolovits}]{jin2020bert}
Di~Jin, Zhijing Jin, Joey~Tianyi Zhou, and Peter Szolovits. 2020.
\newblock Is bert really robust? a strong baseline for natural language attack on text classification and entailment.
\newblock In \emph{Proceedings of the AAAI conference on artificial intelligence}, volume~34, pages 8018--8025.

\bibitem[{Koike et~al.(2024)Koike, Kaneko, and Okazaki}]{koike2024outfox}
Ryuto Koike, Masahiro Kaneko, and Naoaki Okazaki. 2024.
\newblock Outfox: Llm-generated essay detection through in-context learning with adversarially generated examples.
\newblock In \emph{Proceedings of the AAAI Conference on Artificial Intelligence}, volume~38, pages 21258--21266.

\bibitem[{Krishna et~al.(2024)Krishna, Song, Karpinska, Wieting, and Iyyer}]{krishna2024paraphrasing}
Kalpesh Krishna, Yixiao Song, Marzena Karpinska, John Wieting, and Mohit Iyyer. 2024.
\newblock Paraphrasing evades detectors of ai-generated text, but retrieval is an effective defense.
\newblock \emph{Advances in Neural Information Processing Systems}, 36.

\bibitem[{Kukich(1992)}]{kukich1992techniques}
Karen Kukich. 1992.
\newblock Techniques for automatically correcting words in text.
\newblock \emph{ACM computing surveys (CSUR)}, 24(4):377--439.

\bibitem[{Kushnareva et~al.(2024)Kushnareva, Gaintseva, Magai, Barannikov, Abulkhanov, Kuznetsov, Tulchinskii, Piontkovskaya, and Nikolenko}]{kushnareva2024ai}
Laida Kushnareva, Tatiana Gaintseva, German Magai, Serguei Barannikov, Dmitry Abulkhanov, Kristian Kuznetsov, Eduard Tulchinskii, Irina Piontkovskaya, and Sergey Nikolenko. 2024.
\newblock Ai-generated text boundary detection with roft.
\newblock In \emph{1st Conference on Language Modeling (COLM)}, volume 2024.

\bibitem[{Lan et~al.(2019)Lan, Chen, Goodman, Gimpel, Sharma, and Soricut}]{lan2019albert}
Zhenzhong Lan, Mingda Chen, Sebastian Goodman, Kevin Gimpel, Piyush Sharma, and Radu Soricut. 2019.
\newblock Albert: A lite bert for self-supervised learning of language representations.
\newblock \emph{arXiv preprint arXiv:1909.11942}.

\bibitem[{Levenshtein(1966)}]{levenshtein1966binary}
Vladimir~I. Levenshtein. 1966.
\newblock Binary codes capable of correcting deletions, insertions, and reversals.
\newblock \emph{Soviet Physics Doklady}, 10(8):707--710.

\bibitem[{Li et~al.(2020)Li, Ma, Guo, Wang, Qiu, and Tang}]{li2020bert}
Linyang Li, Ruotian Ma, Xiaonan Guo, Qing Wang, Xipeng Qiu, and Xuanjing Tang. 2020.
\newblock Bert-attack: Adversarial attack against bert using bert.
\newblock In \emph{Proceedings of the 2020 Conference on Empirical Methods in Natural Language Processing (EMNLP)}, pages 6193--6202.

\bibitem[{Li et~al.(2021)Li, Ren, and Shi}]{li2021tavat}
Yuan Li, Shuhuai Ren, and Zhouxing Shi. 2021.
\newblock Tavat: Token-aware virtual adversarial training for language understanding.
\newblock In \emph{Proceedings of the 2021 Conference on Empirical Methods in Natural Language Processing}, pages 2924--2936.

\bibitem[{Liang et~al.(2023)Liang, Yuksekgonul, Mao, Wu, and Zou}]{liang2023gpt}
Weixin Liang, Mert Yuksekgonul, Yining Mao, Eric Wu, and James Zou. 2023.
\newblock Gpt detectors are biased against non-native english writers.
\newblock \emph{Patterns}, 4(7).

\bibitem[{Liu et~al.(2024{\natexlab{a}})Liu, Xu, Zhang, Zhang, Ma, Chen, Yu, and Zhang}]{liu2024hqa}
Han Liu, Zhi Xu, Xiaotong Zhang, Feng Zhang, Fenglong Ma, Hongyang Chen, Hong Yu, and Xianchao Zhang. 2024{\natexlab{a}}.
\newblock Hqa-attack: toward high quality black-box hard-label adversarial attack on text.
\newblock \emph{Advances in Neural Information Processing Systems}, 36.

\bibitem[{Liu et~al.(2024{\natexlab{b}})Liu, Liu, Wang, Cheng, Li, Zhang, Lan, and Shen}]{liu2024does}
Shengchao Liu, Xiaoming Liu, Yichen Wang, Zehua Cheng, Chengzhengxu Li, Zhaohan Zhang, Yu~Lan, and Chao Shen. 2024{\natexlab{b}}.
\newblock Does detectgpt fully utilize perturbation? bridging selective perturbation to fine-tuned contrastive learning detector would be better.
\newblock In \emph{Proceedings of the 62nd Annual Meeting of the Association for Computational Linguistics (Volume 1: Long Papers)}, pages 1874--1889.

\bibitem[{Liu et~al.(2022)Liu, Zhang, Wang, Pu, Lan, and Shen}]{liu2022coco}
Xiaoming Liu, Zhaohan Zhang, Yichen Wang, Hang Pu, Yu~Lan, and Chao Shen. 2022.
\newblock Coco: Coherence-enhanced machine-generated text detection under data limitation with contrastive learning.
\newblock \emph{arXiv preprint arXiv:2212.10341}.

\bibitem[{Liu et~al.(2019)Liu, Ott, Goyal, Du, Joshi, Chen, Levy, Lewis, Zettlemoyer, and Stoyanov}]{liu2019roberta}
Yinhan Liu, Myle Ott, Naman Goyal, Jingfei Du, Mandar Joshi, Danqi Chen, Omer Levy, Mike Lewis, Luke Zettlemoyer, and Veselin Stoyanov. 2019.
\newblock Roberta: A robustly optimized {BERT} pretraining approach.
\newblock \emph{arXiv preprint arXiv:1907.11692}.

\bibitem[{Macko et~al.(2024)Macko, Moro, Uchendu, Srba, Lucas, Yamashita, Tripto, Lee, Simko, and Bielikova}]{macko2024authorship}
Dominik Macko, Robert Moro, Adaku Uchendu, Ivan Srba, Jason~Samuel Lucas, Michiharu Yamashita, Nafis~Irtiza Tripto, Dongwon Lee, Jakub Simko, and Maria Bielikova. 2024.
\newblock Authorship obfuscation in multilingual machine-generated text detection.
\newblock \emph{arXiv preprint arXiv:2401.07867}.

\bibitem[{Mitchell et~al.(2023)Mitchell, Lee, Khazatsky, Manning, and Finn}]{mitchell2023detectgpt}
Eric Mitchell, Yoonho Lee, Alexander Khazatsky, Christopher~D Manning, and Chelsea Finn. 2023.
\newblock Detectgpt: Zero-shot machine-generated text detection using probability curvature.
\newblock In \emph{International Conference on Machine Learning}, pages 24950--24962. PMLR.

\bibitem[{Miyato et~al.(2016)Miyato, Maeda, Koyama, and Ishii}]{miyato2016virtual}
Takeru Miyato, Shin{-}ichi Maeda, Masanori Koyama, and Shin Ishii. 2016.
\newblock Virtual adversarial training: A regularization method for supervised and semi-supervised learning.
\newblock \emph{IEEE Transactions on Pattern Analysis and Machine Intelligence}, 41(8):1979--1993.

\bibitem[{Prim(1957)}]{prim1957shortest}
Robert~Clay Prim. 1957.
\newblock Shortest connection networks and some generalizations.
\newblock \emph{The Bell System Technical Journal}, 36(6):1389--1401.

\bibitem[{Radford et~al.(2019)Radford, Wu, Child, Luan, Amodei, and Sutskever}]{radford2019language}
Alec Radford, Jeffrey Wu, Rewon Child, David Luan, Dario Amodei, and Ilya Sutskever. 2019.
\newblock Language models are unsupervised multitask learners.
\newblock \emph{OpenAI Technical Report}, 1(8):1--24.

\bibitem[{Ren et~al.(2019)Ren, Zheng, Chen, Yu, Liu, and Sun}]{ren2019generating}
Shuhuai Ren, Yihe Zheng, Yizhan Chen, Bin Yu, Zhiyuan Liu, and Maosong Sun. 2019.
\newblock Generating natural language adversarial examples through probability weighted word saliency.
\newblock In \emph{Proceedings of the 57th Annual Meeting of the Association for Computational Linguistics}, pages 1085--1097.

\bibitem[{Robbins and Monro(1951)}]{robbins1951stochastic}
Herbert Robbins and Sutton Monro. 1951.
\newblock A stochastic approximation method.
\newblock \emph{The annals of mathematical statistics}, pages 400--407.

\bibitem[{Shi et~al.(2024)Shi, Wang, Yin, Chen, Chang, and Hsieh}]{shi2024red}
Zhouxing Shi, Yihan Wang, Fan Yin, Xiangning Chen, Kai-Wei Chang, and Cho-Jui Hsieh. 2024.
\newblock Red teaming language model detectors with language models.
\newblock \emph{Transactions of the Association for Computational Linguistics}, 12:174--189.

\bibitem[{Siino(2024)}]{siino2024badrock}
M.~Siino. 2024.
\newblock \href {https://aclanthology.org/2024.semeval-1.37} {Badrock at semeval-2024 task 8: Distilbert to detect multigenerator, multidomain and multilingual black-box machine-generated text}.
\newblock In \emph{Proceedings of the 18th International Workshop on Semantic Evaluation}.

\bibitem[{Su et~al.(2023)Su, Zhuo, Wang, and Nakov}]{su2023detectllm}
Jinyan Su, Terry Zhuo, Di~Wang, and Preslav Nakov. 2023.
\newblock Detectllm: Leveraging log rank information for zero-shot detection of machine-generated text.
\newblock In \emph{Findings of the Association for Computational Linguistics: EMNLP 2023}, pages 12395--12412.

\bibitem[{Tiedemann(2012)}]{tiedemann-2012-parallel}
J{\"o}rg Tiedemann. 2012.
\newblock \href {http://www.lrec-conf.org/proceedings/lrec2012/pdf/463_Paper.pdf} {Parallel data, tools and interfaces in {OPUS}}.
\newblock In \emph{Proceedings of the Eighth International Conference on Language Resources and Evaluation ({LREC}'12)}, pages 2214--2218, Istanbul, Turkey. European Language Resources Association (ELRA).

\bibitem[{Uchendu et~al.(2020)Uchendu, Le, Shu, and Lee}]{uchendu2020authorship}
Adaku Uchendu, Thai Le, Kai Shu, and Dongwon Lee. 2020.
\newblock Authorship attribution for neural text generation.
\newblock In \emph{2020 Conference on Empirical Methods in Natural Language Processing, EMNLP 2020}, pages 8384--8395. Association for Computational Linguistics (ACL).

\bibitem[{Verma et~al.(2024)Verma, Fleisig, Tomlin, and Klein}]{verma2024ghostbuster}
Vivek Verma, Eve Fleisig, Nicholas Tomlin, and Dan Klein. 2024.
\newblock Ghostbuster: Detecting text ghostwritten by large language models.
\newblock In \emph{Proceedings of the 2024 Conference of the North American Chapter of the Association for Computational Linguistics: Human Language Technologies (Volume 1: Long Papers)}, pages 1702--1717.

\bibitem[{Wang et~al.(2023)Wang, Li, Ren, Jiang, Zhang, and Qiu}]{wang2023seqxgpt}
Pengyu Wang, Linyang Li, Ke~Ren, Botian Jiang, Dong Zhang, and Xipeng Qiu. 2023.
\newblock Seqxgpt: Sentence-level ai-generated text detection.
\newblock In \emph{Proceedings of the 2023 Conference on Empirical Methods in Natural Language Processing}, pages 1144--1156.

\bibitem[{Wang et~al.(2024{\natexlab{a}})Wang, Feng, Hou, Pu, Shen, Liu, Tsvetkov, and He}]{wang2024stumbling}
Yichen Wang, Shangbin Feng, Abe~Bohan Hou, Xiao Pu, Chao Shen, Xiaoming Liu, Yulia Tsvetkov, and Tianxing He. 2024{\natexlab{a}}.
\newblock Stumbling blocks: Stress testing the robustness of machine-generated text detectors under attacks.
\newblock \emph{arXiv preprint arXiv:2402.11638}.

\bibitem[{Wang et~al.(2024{\natexlab{b}})Wang, Mansurov, Ivanov, Su, Shelmanov, Tsvigun, Whitehouse, Afzal, Mahmoud, Sasaki et~al.}]{wang2024m4}
Yuxia Wang, Jonibek Mansurov, Petar Ivanov, Jinyan Su, Artem Shelmanov, Akim Tsvigun, Chenxi Whitehouse, Osama~Mohammed Afzal, Tarek Mahmoud, Toru Sasaki, et~al. 2024{\natexlab{b}}.
\newblock M4: Multi-generator, multi-domain, and multi-lingual black-box machine-generated text detection.
\newblock In \emph{Proceedings of the 18th Conference of the European Chapter of the Association for Computational Linguistics (Volume 1: Long Papers)}, pages 1369--1407.

\bibitem[{Wang et~al.(2024{\natexlab{c}})Wang, Wang, Liu, Wang, Fu, Lu, Aggarwal, Pei, and Zhou}]{wang2024comprehensive}
Zaitian Wang, Pengfei Wang, Kunpeng Liu, Pengyang Wang, Yanjie Fu, Chang-Tien Lu, Charu~C. Aggarwal, Jian Pei, and Yuanchun Zhou. 2024{\natexlab{c}}.
\newblock A comprehensive survey on data augmentation.
\newblock \emph{arXiv preprint arXiv:2405.09591}.

\bibitem[{Winata et~al.(2023)Winata, Aji, Yong, and Solorio}]{winata2023decades}
Genta Winata, Alham~Fikri Aji, Zheng~Xin Yong, and Thamar Solorio. 2023.
\newblock The decades progress on code-switching research in nlp: A systematic survey on trends and challenges.
\newblock In \emph{Findings of the Association for Computational Linguistics: ACL 2023}, pages 2936--2978.

\bibitem[{Yoo and Qi(2021)}]{yoo2021towards}
Jin~Yong Yoo and Yanjun Qi. 2021.
\newblock Towards improving adversarial training of nlp models.
\newblock In \emph{Findings of the Association for Computational Linguistics: EMNLP 2021}, pages 1001--1013.

\bibitem[{Yu et~al.(2024)Yu, Chen, and He}]{yu2024query}
Zhen Yu, Zhenhua Chen, and Kun He. 2024.
\newblock Query-efficient textual adversarial example generation for black-box attacks.
\newblock In \emph{Proceedings of the 2024 Conference of the North American Chapter of the Association for Computational Linguistics: Human Language Technologies (Volume 1: Long Papers)}, pages 556--569.

\bibitem[{Yuan et~al.(2023)Yuan, Zhang, Chen, and Wei}]{yuan2023bridge}
Lifan Yuan, Yichi Zhang, Yangyi Chen, and Wei Wei. 2023.
\newblock Bridge the gap between cv and nlp! a gradient-based textual adversarial attack framework.
\newblock In \emph{Findings of the Association for Computational Linguistics: ACL 2023}, pages 7132--7146.

\bibitem[{Zamfirescu-Pereira et~al.(2023)Zamfirescu-Pereira, Wong, Hartmann, and Yang}]{zamfirescu2023johnny}
JD~Zamfirescu-Pereira, Richmond~Y Wong, Bjoern Hartmann, and Qian Yang. 2023.
\newblock Why johnny can’t prompt: how non-ai experts try (and fail) to design llm prompts.
\newblock In \emph{Proceedings of the 2023 CHI Conference on Human Factors in Computing Systems}, pages 1--21.

\bibitem[{Zeng et~al.(2023)Zeng, Zheng, Xu, Li, Yuan, and Huang}]{zeng2023certified}
Jiehang Zeng, Xiaoqing Zheng, Jianhan Xu, Linyang Li, Liping Yuan, and Xuanjing Huang. 2023.
\newblock Certified robustness to text adversarial attacks by randomized [mask].
\newblock \emph{Computational Linguistics}, 49(2):345--373.

\bibitem[{Zhang et~al.(2020{\natexlab{a}})Zhang, Williams, Titov, and Sennrich}]{zhang-etal-2020-improving}
Biao Zhang, Philip Williams, Ivan Titov, and Rico Sennrich. 2020{\natexlab{a}}.
\newblock \href {https://doi.org/10.18653/v1/2020.acl-main.148} {Improving massively multilingual neural machine translation and zero-shot translation}.
\newblock In \emph{Proceedings of the 58th Annual Meeting of the Association for Computational Linguistics}, pages 1628--1639, Online. Association for Computational Linguistics.

\bibitem[{Zhang et~al.(2020{\natexlab{b}})Zhang, Zhao, Saleh, and Liu}]{zhang2020pegasus}
Jingqing Zhang, Yao Zhao, Mohammad Saleh, and Peter Liu. 2020{\natexlab{b}}.
\newblock Pegasus: Pre-training with extracted gap-sentences for abstractive summarization.
\newblock In \emph{International conference on machine learning}, pages 11328--11339. PMLR.

\bibitem[{Zhang et~al.(2019)Zhang, Kishore, Wu, Weinberger, and Artzi}]{zhang2019bertscore}
Tianyi Zhang, Varsha Kishore, Felix Wu, Kilian~Q. Weinberger, and Yoav Artzi. 2019.
\newblock Bertscore: Evaluating text generation with bert.
\newblock \emph{arXiv preprint arXiv:1904.09675}.

\bibitem[{Zhang et~al.(2024{\natexlab{a}})Zhang, Hong, Hong, Huang, Wang, Ba, and Ren}]{zhang2024text}
Xinyu Zhang, Hanbin Hong, Yuan Hong, Peng Huang, Binghui Wang, Zhongjie Ba, and Kui Ren. 2024{\natexlab{a}}.
\newblock Text-crs: A generalized certified robustness framework against textual adversarial attacks.
\newblock In \emph{Proceedings of the 2024 IEEE Symposium on Security and Privacy}, pages 53--53.

\bibitem[{Zhang et~al.(2024{\natexlab{b}})Zhang, Yao, Liang, and Xu}]{zhang2024random}
Zeliang Zhang, Wei Yao, Susan Liang, and Chenliang Xu. 2024{\natexlab{b}}.
\newblock Random smooth-based certified defense against text adversarial attack.
\newblock In \emph{Findings of the Association for Computational Linguistics: EACL 2024}, pages 1251--1265.

\bibitem[{Zhong et~al.(2020)Zhong, Tang, Xu, Wang, Duan, Zhou, Wang, and Yin}]{zhong2020neural}
Wanjun Zhong, Duyu Tang, Zenan Xu, Ruize Wang, Nan Duan, Ming Zhou, Jiahai Wang, and Jian Yin. 2020.
\newblock Neural deepfake detection with factual structure of text.
\newblock In \emph{Proceedings of the 2020 Conference on Empirical Methods in Natural Language Processing (EMNLP)}, pages 2461--2470.

\bibitem[{Zhou et~al.(2024)Zhou, He, and Sun}]{zhou2024humanizing}
Ying Zhou, Ben He, and Le~Sun. 2024.
\newblock Humanizing machine-generated content: Evading ai-text detection through adversarial attack.
\newblock In \emph{Proceedings of the 2024 Joint International Conference on Computational Linguistics, Language Resources and Evaluation (LREC-COLING 2024)}, pages 8427--8437.

\end{thebibliography}

\clearpage
\newpage
\appendix
\section{Experimental Setting}
\label{apdx:expset}

\subsection{Dataset}
\label{apdx:dataset}

We used the Semeval Task8 dataset \cite{siino2024badrock} as the primary dataset for training the detector.
This large-scale dataset includes MGTs from ChatGPT 4, Davinci, Bloomz, Dolly, and Cohere, as well as HWTs from WikiHow, Reddit, arXiv, Wikipedia, and PeerRead.
Moreover, the average token length of the dataset is 623.2521. For each scenario, we employed distinct datasets for testing, as detailed in Tabel \ref{tab:exp1}.

\begin{table}[h]
\centering
\renewcommand{\arraystretch}{1.8}
\resizebox{0.45\textwidth}{!}{
\begin{tabular}{ccc}
\hline
\textbf{Scenario}              & \textbf{Dataset}                                      & \textbf{Number of MGTs} \\ \hline
Mixed Edit & Semeval Task8 \cite{siino2024badrock}                                              & 1000  \\ \hline
HMGC & Semeval Task8 \cite{siino2024badrock}                                              & 1000 \\ \hline
Paraphrasing & Semeval Task8 \cite{siino2024badrock}                                              & 1000 \\ \hline
Code-switching          & Semeval Task8 \cite{siino2024badrock}                                              & 1000 \\ \hline
Human Obfuscation  & Semeval Task8 \cite{siino2024badrock}                                              & 1000 \\ \hline
Emoji-cogen           & Wang et al. \cite{wang2024stumbling}                & 500 \\ \hline
Typo-cogen           & Wang et al. \cite{wang2024stumbling}                & 500 \\ \hline
ICL           & Wang et al. \cite{wang2024stumbling}                & 500 \\ \hline
Prompt Paraphrasing            & Wang et al. \cite{wang2024stumbling}                & 500  \\ \hline
CSGen           & Wang et al. \cite{wang2024stumbling}                & 500  \\ \hline
Adversarial Attack             & Semeval Task8 \cite{siino2024badrock}                                              & 500 \\ \hline
\end{tabular}
}
\caption{Experimental scenarios and corresponding datasets.}
\label{tab:exp1}
\end{table}

\subsection{Implementation}
\label{appdx:implementation}
\modelname\ is deployed on a server equipped with 4 NVIDIA A100 GPUs, running on Ubuntu 22.04. The adversarial framework uses the xlm-roberta-base model (279M) as its base detector.
For the evaluation of Mixed Edit Attack, Paraphrasing Attack, Code-switching Attack, and Human Obfuscation, we adopt the concept of "budget" to control the intensity of the attacks for a more fine-grained investigation of model robustness, following the methodology in \citet{wang2024stumbling}.
Specifically, Mixed Edit Attack uses character edit distance \cite{levenshtein1966binary} as the budget, Paraphrasing Attack and Code-switching Attack utilize BERTScore \cite{zhang2019bertscore} as the budget, and Human Obfuscation Attack employs the confusion ratio as the budget.

For all defense methods, we use \texttt{xlm-roberta-base} as the base detector. The sizes of the training set, validation set, and test set are 8000, 1000, and 1000, respectively. The learning rate is set to $1\text{e-}5$, and the number of epochs is fixed at 6.

For all data augmentation-based methods, we use 20\% of MGT for data augmentation. If the method has hyperparameters, they are set according to the reference values in the original paper.
For all adversarial training-based methods, we use 20\% of MGT for adversarial training. If the method has hyperparameters, they are set according to the reference values in the original paper.
For our method, our detector is trained using a label smoothing loss function with a smoothing factor of $\alpha$. We use a trained \texttt{RoBERTa Large} as the Surrogate Model $\mathcal{M}_{\text{sur}}(.)$ due to its strong generalization ability and precise understanding of English text.
Our method's selected hyperparameters are shown in Table~\ref{tab:hyperparameters}.

\begin{table}[h]
    \centering
    \resizebox{0.5\textwidth}{!}{ 
        \begin{tabular}{cc}
            \toprule
            Hyperparameter & Value \\
            \midrule
            Weight Parameter $\lambda$ & 0.05 \\
            Scaling Factor $\epsilon$ & 0.3 \\
            Scaling Factor $\xi$ & 0.01 \\
            Lower Bound of the Uniform Distribution $a$ & 0.5 \\
            Upper Bound of the Uniform Distribution $b$ & -0.5 \\
            Batch Size $M$ & 50 \\
            Epoch $N$ & 6 \\
            Label Smoothing Factor $\alpha$ & 0.1 
            \\
            \bottomrule
        \end{tabular}
    }
    \caption{Hyperparameters for our \modelname.}
    \label{tab:hyperparameters}
\end{table}

\subsection{Evaluation Metrics}
\label{apdx:metric}
We use attack effectiveness metrics and text quality metrics to comprehensively evaluate the performance of defense and attack methods.

\noindent\textit{1) Attack Effectiveness Metrics}

\noindent\textbf{Attack Success Rates (ASR).} The Attack Success Rate (ASR) measures the proportion of successful attacks relative to the total number of attempted attacks.
Note that we only attack text that was detected as machine-written before the attack.
ASR is calculated as follows:
\[
\text{ASR} = \frac{\text{Text detected as HWT after attack}}{\text{Text detected as MGT before attack}}.
\]
For detector, a lower ASR indicates better defense performance, while for adversary, a higher ASR signifies a stronger attack effectiveness.

\noindent\textit{2) Text Quality Metrics}

\noindent\textbf{Perturbation Rate  ($\mathrm{Pert.}$).} $\mathrm{Pert.}$ measures the lexical difference between the adversarial text and the original text. 
It is defined as the ratio of the number of perturbed tokens to the total number of tokens in the text. 
A lower perturbation rate indicates that the adversarial text remains more similar to the original text.

\noindent\textbf{Perplexity Variation ($\Delta\mathrm{PPL}$).} $\Delta\mathrm{PPL}$ measures the change of perplexity, which represents the consistency and fluency of the adversarial examples.
The PPL is calculated as:
\[
\mathrm{PPL} = exp(-\frac{1}{N} \sum_{i=1}^{N} \log P(w_i | w_1, w_2, \dots, w_{i-1})),
\]
where \(N\) is the total number of tokens, \(w_i\) represents the \(i\)-th token, and \(P(w_i | w_1, w_2, \dots, w_{i-1})\) is the probability assigned by the language model. Generally, a lower PPL variation indicates that the quality of the adversarial text is closer to that of the original text. We use GPT-2~\cite{radford2019language} to compute PPL.

\noindent\textbf{Universal Sentence Encoder (USE) score\cite{cer2018universal}. }
USE evaluates the semantic similarity between the adversarial example and the original text.
The USE score is computed as:
\[
\text{USE Score} = \frac{\mathbf{E}_{\text{orig}} \cdot \mathbf{E}_{\text{adv}}}{\|\mathbf{E}_{\text{orig}}\| \|\mathbf{E}_{\text{adv}}\|},
\]
where \( \mathbf{E}_{\text{orig}} \) and \( \mathbf{E}_{\text{adv}} \) are the sentence embeddings of the original and adversarial texts, respectively, generated by the USE model, and \( \|\mathbf{E}\| \) denotes the Euclidean norm. A higher USE score indicates greater semantic similarity. We use USE (Universal Sentence Encoder) to calculate USE score.

\noindent\textbf{Flesch Reading Ease score (\(\Delta r\)) \cite{flesch1948new}.}
$\Delta r$ measures the variation in text readability. The Flesch Reading Ease score is calculated as:

\[
r = 206.835 - 1.015 \times \frac{N_{\text{words}}}{N_{\text{sentences}}} - 84.6 \times \frac{N_{\text{syllables}}}{N_{\text{words}}},
\]
where \( N_{\text{words}} \) is the total number of words in the text, \( N_{\text{sentences}} \) is the total number of sentences, and \( N_{\text{syllables}} \) is the total number of syllables. A smaller \( \Delta r \) indicates less change in readability.

\subsection{Experimental Comparison Methods}
This section provides a detailed introduction to the SOTA methods included in our comparisons.

\subsubsection{Defense Methods}
\label{apdx:defmethod}

\noindent \textbf{Edit Pretraining (EP) \cite{wang2024comprehensive}:} 
An adversarial data augmentation method that blends Mixed Edit Attack into the training set.

\noindent \textbf{Paraphrasing Pretraining (PP) \cite{wang2024comprehensive}:} An adversarial data augmentation method that blends Paraphrasing Attack into the training set.

\noindent
\textbf{CERT-ED \cite{huang2024cert}:}
An adversarial data augmentation method that adapts randomized deletion to effectively safeguard natural language classification models from diverse edit-based adversarial operations, including synonym substitution, insertion, and deletion.

\noindent
\textbf{RanMask \cite{zeng2023certified}:}
An adversarial data augmentation method that randomly masks a proportion of words in the input text, thereby mitigating both word- and character-level adversarial perturbations without assuming prior knowledge of the adversaries’ synonym generation.

\noindent
\textbf{Text-RS \cite{zhang2024random}:}
An adversarial data augmentation method that treats discrete word substitutions as continuous perturbations in the embedding space to reduce the complexity of searching through large discrete vocabularies and bolster the model’s robustness.

\noindent
\textbf{Text-CRS \cite{zhang2024text}:}
An adversarial data augmentation method that is built on randomized smoothing, encompassing various word-level adversarial manipulations—such as synonym substitution, insertion, deletion, and reordering—by modeling them in both embedding and permutation spaces.

\noindent
\textbf{Virtual Adversarial Training (VAT) \cite{miyato2016virtual}:}
An Adversarial Training method that defines the adversarial direction without label information and employs the robustness of the conditional label distribution around each data point against local perturbation as the adversarial loss.

\noindent
\textbf{Token Aware Virtual Adversarial Training (TAVAT) \cite{li2021tavat}:}
An Adversarial Training method that uses token-level accumulated perturbation to better initialize the noise and applies token-level normalization.

\noindent
\textbf{RADAR~\cite{hu2023radar}:} An Adversarial Training method that uses a paraphraser (such as DIPPER) as adversary and a fine-tuned model as detector. The adversary and detector learn from each other and improve the robustness of the detector when facing paraphrasing attack.

\noindent  
\textbf{OUTFOX~\cite{koike2024outfox}:} An adversarial training method that iteratively improves a detector by generating in-context adversarial examples via a co-evolving attacker. The attacker crafts examples to evade the detector, which are then used to strengthen the detector through in-context learning.

\subsubsection{Attack Methods}
\label{sec:attack}
\noindent
\textbf{TextFooler~\cite{jin2020bert}:} 
A black-box attack method targeting text classification and natural language inference tasks. It ranks words by their importance to the model's prediction and replaces them with semantically similar and grammatically correct synonyms to generate adversarial examples while preserving sentence meaning and structure. 

\noindent
\textbf{BERTAttack~\cite{li2020bert}:} 
A black-box attack method that utilizes a pre-trained BERT model to generate adversarial examples. It replaces certain words in the target sentence with high-probability candidates predicted by BERT, ensuring the adversarial examples remain semantically and grammatically similar to the original while misleading the model. 

\noindent
\textbf{PWWS~\cite{ren2019generating}:} 
A black-box attack method designed for text classification models. It computes the saliency of each word in the sentence, ranks them accordingly, and replaces high-saliency words with synonyms from WordNet, generating adversarial examples that mislead the model while preserving meaning. 

\noindent
\textbf{A2T~\cite{yoo2021towards}:} 
A white-box attack method aimed at improving adversarial training. It leverages gradient-based word importance estimation, performing a greedy search from least to most important words, and uses word embedding models or masked language models to generate candidate replacements, ensuring adversarial examples remain semantically similar while misleading the model. 

\noindent
\textbf{HQA~\cite{liu2024hqa}:} 
A black-box adversarial attack method for text classification task. It initializes adversarial examples by minimizing perturbations and iteratively substitutes words using synonym sets to optimize both semantic similarity and adversarial effectiveness while reducing query consumption. 

\noindent
\textbf{ABP~\cite{yu2024query}:} 
A black-box adversarial attack method leveraging prior knowledge to guide word substitutions efficiently. It introduces Adversarial Boosting Preference (ABP) to rank word importance and proposes two query-efficient strategies: a query-free attack (ABPfree) and a guided search attack (ABPguide), significantly reducing query numbers while maintaining high attack success rates. 

\noindent
\textbf{T-PGD~\cite{yuan2023bridge}:} A black-box  zero-query adversarial attack method extending optimization-based attack techniques from computer vision to NLP. It applies perturbations to the embedding layer and amplifies them through forward propagation, then uses a masked language model to decode adversarial examples. 

\noindent
\textbf{FastTextDodger~\cite{hu2024fasttextdodger}:} A black-box adversarial attack designed for query-efficient adversarial text generation. It generates grammatically correct adversarial texts while maintaining strong attack effectiveness with minimal query consumption. 

\subsection{Detailed Information of Text Perturbation Method}
\label{sec:pert}
In this section, we introduce 10 Text Perturbation Methods mentioned in Table~\ref{tab:all_results_modified}. Detailed information of Adversarial Attack Methods can be found in Appendix~\ref{sec:attack}.

\noindent
\textbf{Mixed Edit~\cite{wang2024stumbling}:} A text modification strategy combining homograph substitution, formatting edits, and case conversion to evade detection. It manipulates character representation, encoding, and capitalization to introduce imperceptible variations while preserving readability.
\textbf{Homograph Substitution} exploits visually similar graphemes, characters, or glyphs with different meanings for imperceptible text modifications. We use VIPER~\cite{eger2019text} and Easy Character Embedding Space (ECES) to obtain optimal homoglyph alternatives.
\textbf{Formatting Edits} introduces human-invisible disruptions using special escape characters and format-control Unicode symbols to evade detection. We employ newline (\texttt{\textbackslash n}), carriage return (\texttt{\textbackslash r}), vertical tab (\texttt{\textbackslash v}), zero-width space (\texttt{\textbackslash u200B}), and line tabulation (\texttt{\textbackslash u000B}) to fragment text at the encoding level while preserving visual coherence.
\textbf{Case Conversion} alters letter capitalization within a word by converting uppercase letters to lowercase and vice versa. For example, transforming \texttt{PaSsWoRd} into \texttt{pAsSwOrD} disrupts case-sensitive detection while preserving readability.

\noindent
\textbf{HMGC~\cite{zhou2024humanizing}:} A black-box zero-query attack framework. It uses a surrogate model to approximate the detector, ranks words by gradient sensitivity and PPL, and replaces high-importance words via an encoder-based masked language model. Constraints ensure fluency and semantic consistency, while dynamic adversarial learning refines the attack strategy.

\noindent
\textbf{Paraphrasing~\cite{wang2024stumbling}:} A paragraph-level attack that reorganizes sentence composition to hinder detection. It utilizes Dipper~\cite{krishna2024paraphrasing} to reorder, merge, and split multiple sentences, increasing textual variance while preserving meaning.

\noindent
\textbf{Code-Switching~\cite{winata2023decades}:} A linguistic modification strategy that substitutes words with their synonyms in different languages. It includes a model-free (MF) approach using a static dictionary~\cite{zhang-etal-2020-improving, tiedemann-2012-parallel} for replacements in German, Arabic, or Russian, and a model-required (MR) approach employing the Helsinki-NLP~\cite{helsinki-nlp} model to translate selected words.

\noindent
\textbf{Human Obfuscation:} A semantic alteration technique inspired by Semeval Task8~\cite{siino2024badrock}, where the initial segment of MGT is replaced with an equally long HWT. The confusion ratio measures the extent of content substitution to increase ambiguity.

\noindent
\textbf{Emoji-Cogen~\cite{wang2024stumbling}:} A co-generation attack method that inserts emojis into text generation to perturb the output. Emojis are introduced immediately after a token is sampled, before generating the next token, and are removed post-generation, ensuring natural readability while confusing automated detectors.

\noindent
\textbf{Typo-Cogen~\cite{wang2024stumbling}:} A co-generation attack method that introduces typos during text generation to manipulate lexical structure. Artificial typos are injected into the generated text and subsequently corrected post-generation, preserving overall coherence while disrupting detection models.

\noindent
\textbf{In-Context Learning (ICL)~\cite{wang2024stumbling}:} A prompt attack method designed to produce human-like outputs that evade detection. It provides the generator with a related HWT as a positive example and a vanilla MGT as a negative example, guiding the model to generate more natural and deceptive text.

\noindent
\textbf{Prompt Paraphrasing~\cite{wang2024stumbling}:} A prompt attack method rewriting technique that enhances textual variation while maintaining semantic integrity. It utilizes the Pegasus paraphraser~\cite{zhang2020pegasus} to restructure input prompts.

\noindent
\textbf{Character-Substituted Generation (CSGen)~\cite{wang2024stumbling}:} A prompt attack method that incorporates character substitution strategies within the prompt. The prompt explicitly specifies replacement rules, such as substituting all occurrences of ‘e’ with ‘x’ during generation. For example, given the prompt: \textit{“Continue 20 words with all ‘e’s substituted with ‘x’s and all ‘x’s substituted with ‘e’s: The evening breeze carried a gentle melody…”}, the model generates text following these constraints. A post-processing step then restores the original characters, ensuring a natural final output.

\section{Experiment on Defense}

\subsection{Defense Performance under Different Attack Strengths}

\begin{figure*}[htbp]
    \centering
    \resizebox{1\textwidth}{!}{%
        \includegraphics{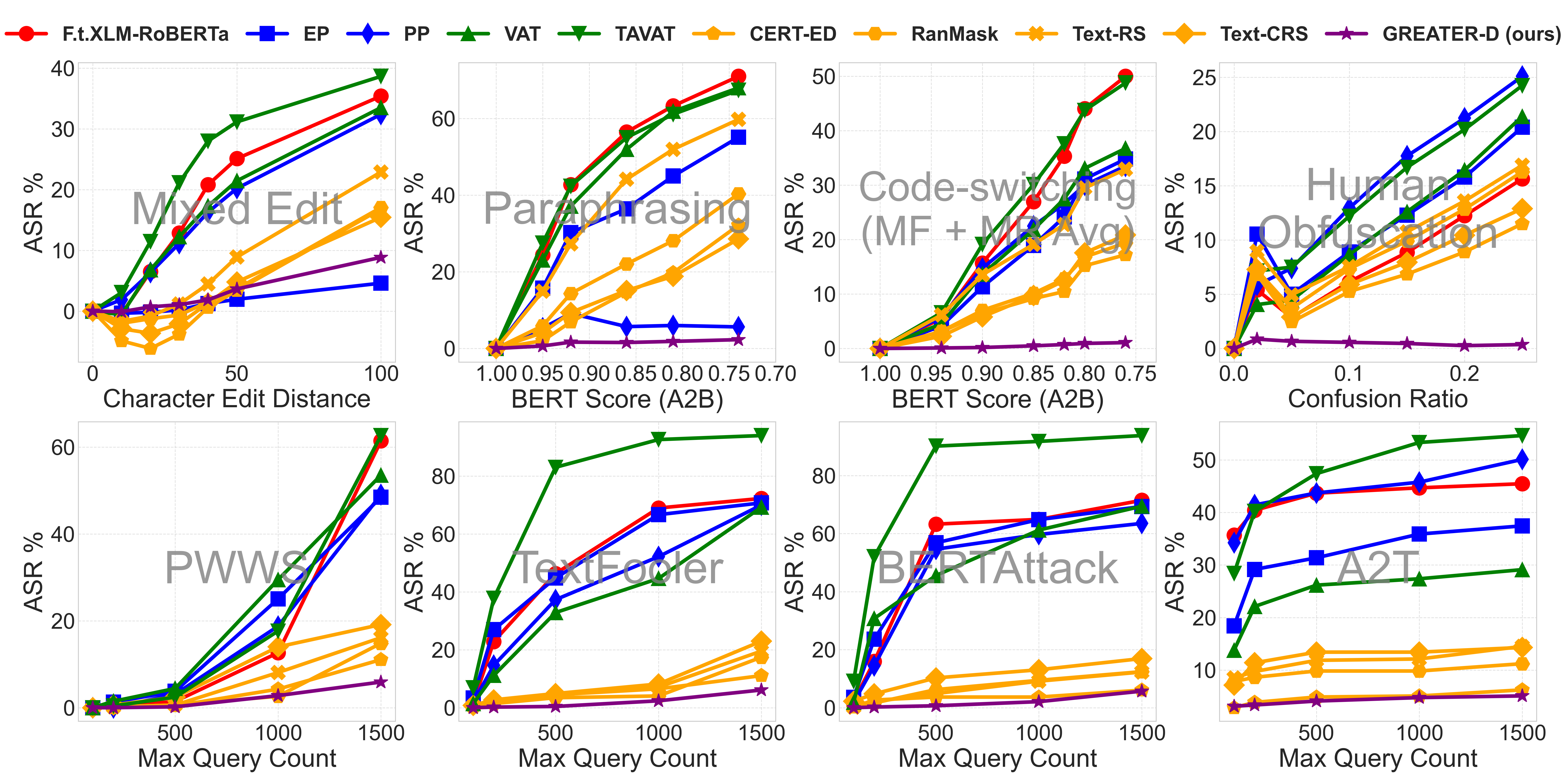} 
    }
    \caption{\textbf{Defense performance under attack with different strengths.} A lower ASR(\%) indicates better defensive performance. A larger character edit distance indicates greater attack intensity in the Mixed Edit Attack. A lower BERT score corresponds to stronger attacks in the Paraphrasing Attack and Code-Switching Attack. A higher obfuscation ratio reflects greater intensity in the Human Obfuscation Attack. Similarly, for PWWS, BERTAttack, TextFooler, and A2T, a larger maximum query count signifies a stronger attack. 
    }
    \label{fig:diffstrength} 
\end{figure*}

We evaluate the resistance of defense methods to increasing attack strengths.
For text perturbation strategies, we employ four text perturbation strategies: Mixed Edit, Paraphrasing, Code-Switching, and Human Obfuscation in the experiment.
Following \citet{wang2024stumbling}, we use Character Edit Distance, BERT Score, BERT Score, and Confusion Ratio as the measure of attack strength for the methods mentioned above, respectively.
For adversarial attacks, we choose PWWS TextFooler, BERTAttack, and A2T to attack MGT detectors, and we utilize max query count to quantify the attack strength.
In the implementation, we change the limit on the attack strength measures to vary the attack intensity.
We show the experimental results in Figure \ref{fig:diffstrength}.

Our experimental results demonstrate that as the attack strength increases, the ASR on our \defensename remains consistently close to zero, whereas other defense methods exhibit significant vulnerabilities under more intensive adversarial attacks.
However, it is worth noticing that under Mixed Edit perturbation, \defensename performs second-best after EP.
This is because EP is originally trained on Mix Edit perturbation and obtains stronger defense against it.
The consistent effective defense against varying attack strengths proves the steadiness of our defense method.

\subsection{Impact of Synchronous Update of \defensename}
The synchronous update mechanism plays a crucial role in enhancing the robustness of the target detector. To explicitly demonstrate its contribution, we ablate the generator-side updates and instead use static texts to update the GREATER-D detector. The results are presented in Table \ref{tab:greater_d_epochs}.
\begin{table}[ht]
\centering
\resizebox{0.45\textwidth}{!}{
\begin{tabular}{c c c}
\toprule
\textbf{Epoch} 
& \makecell{\textbf{GREATER-D} \\ \textbf{(w/o synchronous update)}} 
& \makecell{\textbf{GREATER-D} \\ \textbf{(w/ synchronous update)}} \\
\midrule
1 & 59.04 & 57.64 \\
2 & 52.17 & 46.08 \\
3 & 45.52 & 21.85 \\
4 & 40.17 & 10.28 \\
5 & 34.52 & 4.77 \\
6 & 34.08 & 3.45 \\
\bottomrule
\end{tabular}
}
\caption{\textbf{Defense performance in ASR (\%) of GREATER-D against paraphrasing attacks across different training epochs.}}
\label{tab:greater_d_epochs}
\end{table}

As shown in Table \ref{tab:greater_d_epochs}, the robustness of the GREATER-D detector significantly degrades when it is trained with static adversarial examples. The ASR converges to approximately 34\% after around 5 epochs, which is substantially worse than the result achieved with dynamic updates. This demonstrates the critical role of the dynamic update mechanism in enhancing model robustness.

\subsection{Defense Performance of Different Epochs}

To demonstrate the effectiveness of the adversarial training procedure, we analyze the GREATER-A and GREATER-D performance after each training round. The detailed analysis of training dynamics is in the Table~\ref{tab:greater_a_d_asr}. We present the ASR of GREATER-A and defense performance of GREATER-D against Paraphrasing Attack at each epoch from 1 to 6. We find that as training progresses, the GREATER-A becomes stronger at attacking XLM-Roberta detector. Moreover, GREATER-D exhibits more robust defense performance on the attack it is not trained with.

\begin{table}[ht]
\centering
\resizebox{0.45\textwidth}{!}{
\begin{tabular}{c c c}
\toprule
\textbf{Epoch} & \textbf{GREATER-A ASR} ↑ & \textbf{GREATER-D ASR} ↓ \\
\midrule
1 & 77.04 & 57.64 \\
2 & 79.03 & 46.08 \\
3 & 84.61 & 21.85 \\
4 & 88.94 & 10.28 \\
5 & 94.06 & 4.77 \\
6 & 96.58 & 3.45 \\
\bottomrule
\end{tabular}
}
\caption{ASR (\%) of the Adversary against F.t.XLM-RoBERTa-Base detector and the Detector against Paraphrasing Attack across different epochs.}
\label{tab:greater_a_d_asr}
\end{table}

\subsection{Experiments on Multilingual data}

We evaluate the multilingual performance of different defense methods using German and Urdu datasets~\cite{siino2024badrock}, and the results are shown in Table \ref{tab:multilingual_defense}. As observed from the table, our model achieves strong performance on both languages and yields the best overall results, indicating its effectiveness in multilingual detection scenarios. 

\begin{table*}[ht]
\centering
\resizebox{1\textwidth}{!}{
\begin{tabular}{c c c c c c c c c c c c}
\toprule
\textbf{Source} & \textbf{Metric} & \textbf{F.t.XLM-RoBERTa} & \textbf{EP} & \textbf{PP} & \textbf{VAT} & \textbf{TAVAT} & \textbf{CERT-ED} & \textbf{RanMask} & \textbf{Text-RS} & \textbf{Text-CRS} & \textbf{GREATER-D} \\
\midrule
\multirow{2}{*}{German} 
& \textit{Acc} & 68.30 & 59.84 & \cellcolor{green!50}{73.06} & 67.44 & 59.37 & 55.60 & 61.13 & 65.04 & 56.80 & 72.86 \\
& \textit{F1}  & 75.36 & 71.13 & \cellcolor{green!50}{78.40} & 75.30 & 71.08 & 68.90 & 71.40 & 73.51 & 69.55 & 77.76 \\
\midrule
\multirow{2}{*}{Urdu} 
& \textit{Acc} & 56.01 & 55.09 & 56.88 & 63.33 & 59.04 & 66.60 & 64.03 & \cellcolor{green!50}{70.62} & 70.21 & 66.24 \\
& \textit{F1}  & 65.27 & 69.14 & 69.85 & 69.07 & 68.56 & 74.97 & 69.81 & 76.10 & \cellcolor{green!50}{76.21} & 74.89 \\
\midrule
\multirow{2}{*}{Overall} 
& \textit{Acc} & 62.16 & 57.47 & 64.97 & 65.39 & 59.21 & 61.10 & 62.58 & 67.83 & 63.51 & \cellcolor{green!50}{69.55} \\
& \textit{F1}  & 70.32 & 70.14 & 74.13 & 72.19 & 69.82 & 71.94 & 70.61 & 74.81 & 72.88 & \cellcolor{green!50}{76.33} \\
\bottomrule
\end{tabular}
}
\caption{Performance of different defense methods in a multilingual setting. The best result in each group is highlighted with a \colorbox{green!50}{green} background.}
\label{tab:multilingual_defense}
\end{table*}

\section{Evaluations on \attackname Performance}

\subsection{Attack on the real-world detector}

Speaking of the fact that the adversary and target model share the same training dataset (but different sets), it is necessary to evaluate if the attack conducted by \attackname would generalize to other detectors which trained on different datasets. To validate this, we attack a close-sourced commercial detector GPTZero\footnote{\url{https://gptzero.me/}} with \attackname, which is a strict black-box setting because no other information except for the input and output of the target model is available. The results are in the Table \ref{tab:attack_performance_comparison}. It demonstrates that our \attackname remains high ASR compared with other SOTA methods, though we are not aware of the architecture, parameters, or training data at all.

\begin{table}[ht]
\centering
\renewcommand\arraystretch{1.4}
\resizebox{0.45\textwidth}{!}{
\begin{tabular}{lcccccc}
\toprule
\textbf{Method} 
& \textbf{Avg Queries} $\downarrow$ 
& \textbf{ASR (\%)} $\uparrow$ 
& \textbf{Pert. (\%)} $\downarrow$
& \textbf{\(\Delta\mathrm{PPL}\)} $\downarrow$
& \textbf{USE} $\uparrow$
& \textbf{\(\Delta r\)} $\downarrow$\\
\toprule

HQA 
& 172.38 
& 63
& 11.23 
& 133.58 
& 0.8312 
& 66.24 \\

FastTextDodger 
& 162.58 
& 70
& 10.58
& 96.25 
& 0.9096 
& 42.15 \\

ABP 
& 133.46
& 82 
& 9.13 
& 72.78 
& 0.8613 
& 36.18 \\

T-PGD 
& 189.36 
& 42
& 13.12 
& 168.54 
& 0.8077 
& 60.51 \\

\attackname (ours) 
& \cellcolor{green!50}42.57 
& \cellcolor{green!50}100
&  \cellcolor{green!50} 7.82 
& \cellcolor{green!50} 67.13 
& \cellcolor{green!50} 0.9393 
& \cellcolor{green!50} 13.39 \\

\bottomrule
\end{tabular}
}
\caption{\textbf{The attack results of different attack methods on the GPTZero detector.} For each sample, the maximum number of queries is limited to 200. The best result in each group is highlighted with a \colorbox{green!50}{green} background.}
\label{tab:attack_performance_comparison}
\end{table}

\subsection{Generalization to Other Detector}

To demonstrate the generalization of our \attackname, we conduct additional experiments with two zero-shot detectors: Fast-DetectGPT \cite{bao2023fast} and Binoculars \cite{hans2024spotting}. Moreover, to further demonstrate the generalization ability of \attackname, we also test it on a decoder-only detector F.t.GPT2~\cite{radford2019language} which shares a different architecture with XLM-RoBERTa. As shown in the Table~\ref{tab:asr_detectors}, \attackname outperforms all comparison methods when attacking zero-shot and decoder-only detectors, which demonstrates the universal effectiveness of \attackname.

\begin{table}[ht]
\centering
\renewcommand\arraystretch{1.4}
\resizebox{0.45\textwidth}{!}{
\begin{tabular}{l c c c}
\toprule
\textbf{Method} & \textbf{Fast-DetectGPT} & \textbf{Binoculars} & \textbf{F.t.GPT2} \\
\midrule
WordNet           & 51.36 & 37.35 & 42.66 \\
Back Translation  & 37.50 & 19.13 & 1.26 \\
Rewrite           & 10.97 & 8.88 & 30.02 \\
T-PGD             & 56.83 & 41.91 & 70.46 \\
HMGC              & 17.39 & 31.43 & 35.08 \\
\attackname (ours)         & \cellcolor{green!50}61.14 & \cellcolor{green!50}65.83 & \cellcolor{green!50}74.25 \\
\bottomrule
\end{tabular}
}
\caption{\textbf{The ASR (\%) across different detectors under different attack methods.} The best result in each group is highlighted with a \colorbox{green!50}{green} background.}
\label{tab:asr_detectors}
\end{table}

\section{Ablation Study}

\begin{table}[ht]
\centering
\renewcommand\arraystretch{1.4}
\resizebox{0.45\textwidth}{!}{
\begin{tabular}{ccccccc}
\toprule
\textbf{Method} & \textbf{Avg Queries $\downarrow$} & \textbf{ASR (\%) $\uparrow$} & \textbf{Pert. (\%) $\downarrow$} & \textbf{$\Delta$PPL $\downarrow$} & \textbf{USE $\uparrow$} & \textbf{$\Delta$r $\downarrow$} \\ 
\toprule

\textbf{R+NP} & \cellcolor{green!50}26.61 & 75.77 & 11.56 & 106.29 & 0.9324 & 14.41 \\

\textbf{R+P} & 53.22 & 75.77 & 9.51 & 58.72 & 0.9497 & 12.13 \\

\textbf{S+NP} & 31.32 & \cellcolor{green!50}96.58 & 11.28 & 85.18 & 0.9136 & 21.80 \\

\midrule

\textbf{Mask-T} & 303.82 & 96.38 & 8.33 & \cellcolor{green!50}27.12 & \cellcolor{green!50}0.9696 & \cellcolor{green!50}4.73 \\

\textbf{GREATER-W} & 33.21 & 96.38 & 11.49 & 65.16 & 0.9482 & 10.04 \\

\textbf{GREATER-WordNet} & 28.41 & 80.89 & 16.68 & 53.82 & 0.9199 & 13.21 \\

\textbf{GREATER-A} & 62.63 & \cellcolor{green!50}96.58 & \cellcolor{green!50}7.26 & 35.22 & 0.9506 & 9.21 \\
\bottomrule

\end{tabular}
}
\caption{\textbf{Ablation study on the \attackname.}
The best result in each group is highlighted with a \colorbox{green!50}{green} background.}
\label{tab:ablation_study}

\end{table}

To demonstrate the effectiveness of each component in the design of \defensename and \attackname, we conduct ablation experiments. 
The ablation models are as follows:

\noindent\textbf{R+NP:} Randomly select tokens to perturb and not apply pruning.

\noindent\textbf{R+P:} Randomly select tokens to perturb and apply pruning.

\noindent\textbf{S+NP:} Select tokens to perturb with the important token identification module but not apply pruning.

\noindent\textbf{Mask-T:} Select tokens to perturb with the important token identification but mask them instead of adding perturbation.

\noindent\textbf{GREATER-W:} Select words to perturb with the important token identification module instead of tokens.

\noindent\textbf{GREATER-WordNet:} Select words to perturb with the important token identification module and substitute them with synonyms.

As shown in Table~\ref{tab:ablation_study}, the original \attackname exhibits the most balanced and effective performance in generating adversarial examples.
Comparing \textbf{R+P} to \attackname, we find the ASR drops by \textbf{20.81}, indicates that identifying important tokens is of great significance for the attack to be successful.
Moreover, greedy pruning is important for maintaining the text quality of adversarial examples.
\textbf{S+NP} achieves the same ASR with \attackname but is left far behind in terms of PPL, USE, and $\Delta$r.
The perturbation strategy also plays a crucial role in the adversarial attack.
Masking the important token instead of perturbing the embedding greatly increases the number of queries.
Applying perturbation on the word level or substituting important words with synonyms also degrades the performance of the adversary.

\section{Mathematical Analysis of Perturbation Rate and Query Complexity}
\label{sec:analysis}

In this section, we provide the theoretical analysis of our proposed adversarial example generation framework, focusing on two crucial metrics: \textit{(i)} the \textbf{perturbation rate}, which characterizes the fraction of modified tokens in an adversarial example; and \textit{(ii)} the \textbf{query complexity}, which measures the number of queries made to the target detector in a black-box setting. We establish strict upper and lower bounds on both metrics, illustrating the efficiency and effectiveness of our method.

\subsection{Perturbation Rate Analysis}

\noindent
\begin{definition}
\label{def:randomvars}
Let \( Z \) denote the maximum number of tokens allowed to be modified. Let \(P\) be the total number of tokens (out of \(T\)) perturbed by \emph{greedy search}, taking integer values in the interval \([1, Z]\), where \(1 \le P \le Z\). Let \(Y\) be the fraction of those \(P\) tokens retained (not reverted) by \emph{greedy pruning}, taking real values in \([0, 1]\). Formally, 
\[
  P \;\in\;\{1,2,\dots,Z\},
  \qquad
  Y \;\in\;[\,0,\,1\,].
\]
We assume \(P\) and \(Y\) exhibit \emph{weak dependence}, meaning they have a small but nonzero covariance \(\mathrm{Cov}(P, Y)\).
\end{definition}

\begin{definition}
\label{def:perturbation_rate}
Let \(P \cdot Y\) be the expected count of tokens that remain modified. The \emph{perturbation rate $\rho$} is defined as
\[
  \rho = \frac{P \cdot Y}{T},
\]
where \(T\) is the length of the original text.
\end{definition}

\begin{theorem}
\label{thm:bounds1}
Let \(P\) and \(Y\) modeled as truncated normal random variables on \([1,Z]\) and \([0,1]\), respectively. Then under mild assumptions on the truncation intervals, we have
\[
  \mathbb{E}[\,P\,] \;\approx\; \frac{Z+1}{2}
  \quad\text{and}\quad
  \mathbb{E}[\,Y\,] \;\approx\; \frac{1}{2}.
\]
\end{theorem}

\begin{proof}
\textbf{(1) Truncated Normal for \(P\) and \(Y\).}  
In adversarial text attacks, \(P\) often emerges from an aggregation of (approximately) Bernoulli decisions: each of the \(T\) tokens has a non-negligible probability of being deemed “important,” subject to a global cap of \(Z\). By the Central Limit Theorem (CLT), this sum is close to normally distributed with mean \(\mu_P\) and variance \(\sigma_P^2\). Because \(P\) cannot exceed \(Z\) (and must be at least \(1\) to induce misclassification), we say
\[
  P \;\sim\; \mathcal{N}_t(\mu_P,\;\sigma_P^2;\;1,\;Z),
\]
where \(\mathcal{N}_t(\cdot)\) denotes a normal distribution truncated to the integer range \([1,Z]\). Analogously, once \(P\) tokens are selected, \emph{greedy pruning} decides to keep or revert each token, again creating a sum of i.i.d.\ Bernoulli-like indicators. Dividing by \(P\) yields
\[
  Y 
  \;=\; \frac{\text{number of retained tokens}}{P}
\]
\[
  \;\sim\; \mathcal{N}_t(\mu_Y,\;\sigma_Y^2;\;0,\;1),
\]
a (truncated) normal over \([0,1]\).

\noindent
\textbf{(2) Standard Formulas for Truncated Normal Means.}  
From truncated normal theory, if \(X\sim \mathcal{N}(\mu,\sigma^2)\) is restricted to \([a,b]\), then its mean is
\[
  \mathbb{E}[\,X\,]
  \;=\;
  \mu 
  \;+\;
  \sigma\,\frac{\phi\!\bigl(\tfrac{a-\mu}{\sigma}\bigr) - \phi\!\bigl(\tfrac{b-\mu}{\sigma}\bigr)}
                     {\Phi\!\bigl(\tfrac{b-\mu}{\sigma}\bigr) - \Phi\!\bigl(\tfrac{a-\mu}{\sigma}\bigr)},
\]
where \(\phi\) and \(\Phi\) are the standard normal PDF and CDF, respectively. Thus, for
\(
  P \sim \mathcal{N}_t(\mu_P,\sigma_P^2;\,1,Z),
\)
\(\mathbb{E}[P]\) depends on how far \(\mu_P\) is from the boundaries \(\{1,Z\}\). Similarly, \(\mathbb{E}[Y]\) depends on truncation at \([0,1]\).

\noindent
\textbf{(3) Approximate Symmetry in Practice.}  
For symmetric truncated normal distributions \( P \) and \( Y \), it is evident that the following expectations hold:  
\[
\mu_P = \frac{Z+1}{2}, \quad \mu_Y = \frac{1}{2}.
\]
So long as \(\mu_P\) and \(\sigma_P\) ensure that \(\frac{1-\mu_P}{\sigma_P}\) and \(\frac{Z-\mu_P}{\sigma_P}\) are not extreme, the truncation does not drastically shift the mean from \(\mu_P\). Concretely, if \(Z\) is sufficiently large and \(\mu_P\approx \tfrac{Z+1}{2}\), then
\[
  \mathbb{E}[\,P\,]
  \;\approx\;
  \mu_P
  \;\approx\;
  \frac{Z+1}{2},
\]
Likewise, if \(\mu_Y \approx \tfrac12\) and \(\sigma_Y\) is moderate, \(\mathbb{E}[\,Y\,]\approx \tfrac{1}{2}\) despite the boundaries \([0,1]\).

\noindent \textbf{Combining (1), (2), and (3), Theorem~\ref{thm:bounds1} is proved.}
\end{proof}

\begin{theorem}
\label{thm:bounds}
Let \( T \) be the length of the original text and \( M \) be the maximum number of tokens that can be perturbed. Then, the perturbation rate \( \rho \) satisfies \( \frac{1}{T} \leq \rho \leq \frac{Z}{T} \), and its expected value is approximately \( \mathbb{E}[\rho] \approx \frac{Z}{4T} \).
\end{theorem}

\begin{proof}
\textbf{(1) Basic Bounds.}  
Since \(P\) is at least 1 and at most \(M\), and \(Y\) is at least 0 and at most 1, the product \(P \cdot Y\) satisfies
\[
  1 \;\le\; P \;\le\; Z,
  \qquad
  0 \;\le\; Y \;\le\; 1
\]
\[
  \quad\Longrightarrow\quad
  1 \;\le\; P \cdot Y \;\le\; Z.
\]

Dividing by \(T\) yields the strict bounds
\[
  \frac{1}{T}
  \;\le\;
  \frac{P \cdot Y}{T}
  \;\le\;
  \frac{Z}{T}.
\]
The lower limit \(1/T\) reflects that at least one token must change to induce a misclassification; the upper limit \(\tfrac{Z}{T}\) follows from the maximal \(Z\) token modifications.

\noindent
\textbf{(2) Expected Product with Weak Dependence.}  
By definition, the covariance between \(P\) and \(Y\) is
\[
  \mathbb{E}[\,P\,Y\,]
  \;=\;
  \mathbb{E}[\,P\,]\;\mathbb{E}[\,Y\,]
  \;+\;
  \mathrm{Cov}(P, Y).
\]
When the detector is relatively robust, \emph{greedy search} skews \(P\) toward larger values (close to \(M\)), and pruning skews \(Y\) toward retention (close to 1). In that case, \(\mathbb{E}[\,P\,]\) is large, \(\mathbb{E}[\,Y\,]\) is near 1, and a small positive covariance implies
\[
  \mathbb{E}[\,P \cdot Y\,]
  \;>\;
  \mathbb{E}[\,P\,]\;\mathbb{E}[\,Y\,].
\]
However, the final mean number of changed tokens cannot exceed \(M\). Conversely, if the detector is weak, \(\mathbb{E}[\,P\,]\) may approach 1, and \(\mathbb{E}[\,Y\,]\) might be relatively modest, implying 
\[
  \mathbb{E}[\,P\,Y\,]
  \;<\;
  \mathbb{E}[\,P\,]\;\mathbb{E}[\,Y\,]
  \;+\;
  |\mathrm{Cov}(P, Y)|.
\]
In all cases, the weak dependence ensures \(\mathrm{Cov}(P, Y)\) is bounded such that 
\[
  1
  \;\le\;
  P \cdot Y
  \;\le\;
  Z,
\]
preventing the expected perturbation rate from falling below \(\tfrac{1}{T}\) or above \(\tfrac{Z}{T}\).

\noindent
\textbf{(3) Characteristic Mean \(\frac{Z}{4T}\).}  
By Theorem~\ref{thm:bounds1}, we have \(\mathbb{E}[\,P\,]\approx \frac{Z+1}{2}\) and \(\mathbb{E}[\,Y\,]\approx \frac{1}{2}\). If \(\mathrm{Cov}(P, Y)\) is small or near zero, then
\[
  \mathbb{E}[\,P \cdot Y\,]
  \;=\;
  \mathbb{E}[\,P\,]\;\mathbb{E}[\,Y\,]
  \;+\;
  \mathrm{Cov}(P, Y)
\]
\[
  \;\approx\;
  \frac{Z+1}{2}\,\times\,\frac{1}{2}
  \;=\;
  \frac{Z+1}{4}.
\]
For large \(Z\), \(\frac{Z+1}{4T}\approx \frac{Z}{4T}\). If \(\mathrm{Cov}(P, Y)\) modestly raises or lowers this sum, the final mean \(\mathbb{E}[\,P\,Y\,]\) still cannot breach the fundamental \(\bigl[1,\,Z\bigr]\) interval. This shows that \(\frac{Z}{4T}\) is a natural approximate pivot for the average perturbation rate under moderate parameters. 

In Section \ref{sec:GeneratorResult}, we set \( Z = 0.3T \). If the number of iterations exceeds this upper bound, the attack is considered a failure. Therefore, the expected perturbation rate is\[
\frac{0.3T}{4T} \approx 7.5\%,\] which is very close to the result obtained in the Section \ref{sec:GeneratorResult} (7.26\%).

\noindent \textbf{Combining (1), (2), and (3), Theorem \ref{thm:bounds} is proved.}
\end{proof}

Note that if the detector forces \emph{greedy search} to repeatedly fail early (producing a right-skewed \(P\)-distribution concentrated near \(Z\)) and \emph{greedy pruning} is left-skewed (so that \(\mathbb{E}[\,Y\,]\) is close to 1), then the mean \(\mathbb{E}[\,P \cdot Y\,]\) exceeds \(\tfrac{Z+1}{4}\), but still cannot exceed \(Z\). Conversely, if \emph{greedy search} finds success quickly, giving a left-skewed \(P\)-distribution near 1, then \(\mathbb{E}[\,P\,]\) might drop below \(\tfrac{Z+1}{2}\) but never below 1; similarly, if \emph{pruning} is so aggressive that \(\mathbb{E}[\,Y\,]\) is significantly below \(\tfrac{1}{2}\), the mean also decreases. Consequently, the expected perturbation rate remains strictly within the \(\bigl[\tfrac{1}{T},\;\tfrac{Z}{T}\bigr]\) interval, but its exact value depends on the interplay of these skewed distributions. For moderate skewness on both \(P\) and \(Y\), \(\tfrac{Z}{4T}\) is a characteristic outcome.

\subsection{Query Complexity Analysis}

\noindent
\begin{definition}
\label{def:queries}
Let \(Q_{\mathrm{G}}\) be the total number of queries made by \emph{greedy search} and \(Q_{\mathrm{P}}\) be the total number made by \emph{greedy pruning}. Define the overall \emph{query complexity}:
\[
  Q \;=\; Q_{\mathrm{G}} + Q_{\mathrm{P}}.
\]
\end{definition}

\begin{theorem}
\label{thm:query_complexity_bounds}
Let \(T\) be the length of the text, and let \(Z = 0.3\,T\) be the maximum iteration count for both \emph{greedy search} and \emph{greedy pruning}. Then the total number of queries \(Q\) used to construct one adversarial example lies within \(1 \leq Q \leq 2Z\), which implies \(Q = O(T)\) and \(Q = \Omega(1)\). Moreover, if the \emph{average} perturbation rate is relatively low, then the number of required iterations typically shrinks, resulting in a correspondingly smaller \(Q\).
\end{theorem}

\begin{proof}
\textbf{(1) Basic Range of Queries.}  
The minimal number of queries is \(1\), occurring if \emph{greedy search} succeeds in its very first attempt, requiring no \emph{greedy pruning}. Conversely, if both \emph{greedy search} and \emph{greedy pruning} use their maximum of \(Z\) single-query iterations, we have
\[
  Q = Q_{\mathrm{G}} + Q_{\mathrm{P}}
  \;\le\; Z + Z
  \;=\; 2\,Z.
\]
Since \(Z=0.3\,T\), \(Q\le 2\,Z\) translates to \(Q = O(T)\). The trivial lower bound of 1 implies \(Q=\Omega(1)\).

\noindent
\textbf{(2) Perturbation Rate and Query Trade-Off.}  
Let \(\rho\) denote the fraction of tokens altered in a final adversarial example, as described in the perturbation rate analysis. A \emph{smaller} \(\rho\) typically indicates that the detector is fooled by changing fewer tokens, implying fewer iterative steps in \emph{greedy search}. Thus, a lower \(\rho\) correlates with \emph{fewer} queries: once the necessary (small) set of tokens is found, misclassification is often achieved without exhausting all \(Z\) iterations. On the other hand, a higher \(\rho\) suggests more alterations, potentially requiring more query rounds to finalize a successful attack.

\noindent
\textbf{(3) Expected Complexity under Moderate Perturbation.}  
From the previous analysis, the average perturbation rate can hover around \(\tfrac{Z}{4T}\) under typical conditions. This moderate \(\rho\) implies that \emph{greedy search} rarely needs the full \(Z\) steps to identify the required modifications, and \emph{greedy pruning} likewise terminates without iterating over all \(Z\). Consequently, the \emph{expected} \(Q\) is substantially below the worst-case \(2\,Z\). Formally, we have
\[
  \mathbb{E}[\,Q_{\mathrm{G}}\,] \;<\; Z,
  \quad
  \mathbb{E}[\,Q_{\mathrm{P}}\,] \;<\; Z,
\]
\[
  \Longrightarrow\quad
  \mathbb{E}[\,Q\,] \;=\; \mathbb{E}[\,Q_{\mathrm{G}}\,] + \mathbb{E}[\,Q_{\mathrm{P}}\,]
  \;<\; 2\,Z.
\]

thus still preserving \(O(T)\) upper complexity while being strictly smaller on average. 

\noindent \textbf{Combining (1), (2), and (3), Theorem \ref{thm:query_complexity_bounds} is proved.}
\end{proof} 
\section{Adversarial Training Process}
\label{appd: pseudo}
Algorithm \ref{alg:algorithm1_short} shows the detailed optimization process of \attackname and \defensename. The two components are updated alternatively in the same training step.

\begin{algorithm}[h]
\small
\caption{Adversarial Training Procedure}
\label{alg:algorithm1_short}
\begin{algorithmic}[1]
        \State \textbf{Input:} Training set $D_{train}$, surrogate detector $\mathcal{M}_{sur}$.

    \begin{center}
            \textit{**** training phase begins ****}
        \end{center}
    \State \textbf{Initialize:} target detector $\mathcal{M}_{tar}$, importance scoring network $\mathcal{F}_{\theta}$ and $epoch \leftarrow 0$.
    \While{$epoch< epoch_{max}$}
       \For{each batch samples $\{X_i, c_i\}_{i=0}^{N}$ in $D_{train}$}
       \State $D_{adv}\leftarrow \{\}$
       \For{each MGT in $\{X_i, c_i\}_{i=0}^{N}$}
       \State Obtain the last layer hidden state by Eq.\eqref{eq:1}.
       \State Obtain the importance score by Eq.\eqref{eq:2}.
       \State Construct the Important-token Set $\mathbf{I}$ by Eq.\eqref{eq:3}.
       \State Execute the Greedy Search Algorithm\ref{alg:algorithm2_short}.
       \State Execute the Greedy Pruning Algorithm\ref{alg:algorithm3_short}.
       \State{Get the adversarial example:}
       \begin{center}
            $D_{adv} \gets D_{adv} \cup \{X_i, c_i\}_{i=0}^{N}$.
        \end{center}
       \EndFor
       \State $D_{adv} \gets D_{adv} \cup \{\tilde{X}_i, c_i\}$.
       \State Classify $D_{adv}$ via $\mathcal{M}_{tar}$ and obtain the output.
       \State Calculate loss $\mathcal{L}_{\text{A}}$ by Eq.\eqref{eq:14}.
       \State Calculate loss $\mathcal{L}_{\text{D}}$ by Eq.\eqref{eq:9}.
       \State Update $\mathcal{M}_{sar}$ and $\mathcal{F}_{\theta}$ via SGD \shortcite{robbins1951stochastic}.
       \EndFor 
    \EndWhile
    \State \textbf{Output:} Trained detector $\mathcal{M}_{tar}$ and $\mathcal{F}_{\theta}$.
\end{algorithmic}
\end{algorithm}

\section{Case Study}

Table \ref{tab:case_study_results} presents a case study of our adversary \attackname, in which our approach \attackname outperforms other SOTA methods regarding semantic preservation and reduction in perturbation rate.

\begin{table*}[!htbp]
\centering
\renewcommand\arraystretch{1.4}
\resizebox{1\textwidth}{!}{

\begin{tabular}{lm{0.6\textwidth}>{\centering\arraybackslash}m{0.2\textwidth}>{\centering\arraybackslash}m{0.08\textwidth}>{\centering\arraybackslash}m{0.08\textwidth}}
\hline
\textbf{Method} & \textbf{Text} & \textbf{Result} & \textbf{Pert. (\%)} & \textbf{Queries} \\ 
\hline

\textbf{Original MGT} & Suggested by the Scottish Parliamentary Constituencies Commission in 2003, and adopted at Holyrood on 1 May 2004. Area of Scotland's 32nd largest council area - covering parts of East Renfrewshire Council Area & \makecell{Machine-written \\ \includegraphics[width=0.5cm]{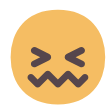}} & - & - \\

\hline

\textbf{PWWS} & Suggested by the Scottish Parliamentary Constituencies Commission in 2003, and \textcolor{red}{adoptions} at Holyrood on 1 \textcolor{red}{Probability} 2004. Area of Scotland's 32nd \textcolor{red}{highest} council area - covering \textcolor{red}{item} of East Renfrewshire Council Area & Human-written (Succeeded) \makebox[-1pt][l]{\includegraphics[width=0.5cm]{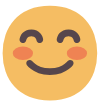}} & 11.76 & 107 \\

\hline

\textbf{TextFooler} & Suggested by the Scottish Parliamentary Constituencies Commission in 2003, and adopted at Holyrood on 1 May 2004. Area of Scotland's 32nd largest council area - covering parts of East Renfrewshire Council Area & Machine-written (Failed) \makebox[-1pt][l]{\includegraphics[width=0.5cm]{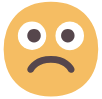}} & - & - \\

\hline

\textbf{BERTAttack} & Suggested by the Scottish \textcolor{red}{rural} Constituencies Commission in 2003, and \textcolor{red}{abolished} at Holyrood on 1 May 2004. Area of \textcolor{red}{ward} 32nd \textcolor{red}{most} council \textcolor{red}{constituency} - \textcolor{red}{almost} \textcolor{red}{all} of East Renfrewshire Council Area & Human-written (Succeeded) \makebox[-1pt][l]{\includegraphics[width=0.5cm]{genshin/smiley.png}} & 20.59 & 60 \\

\hline

\textbf{A2T} & Suggested by the Scottish Parliamentary Constituencies Commission in 2003, and adopted at Holyrood on 1 May 2004. Area of Scotland's 32nd largest council area - covering parts of East Renfrewshire Council Area & Machine-written (Failed) \makebox[-1pt][l]{\includegraphics[width=0.5cm]{genshin/frownie.png}} & - & - \\

\hline

\textbf{FastTextDodger} & Suggested by the Scottish Parliamentary Constituencies Commission in 2003\textcolor{red}{:} and adopted at Holyrood on 1 May 2004. Area \textcolor{red}{of' 's} 32nd largest council area - covering parts of East \textcolor{red}{council Council} & Human-written (Succeeded) \makebox[-1pt][l]{\includegraphics[width=0.5cm]{genshin/smiley.png}} & 14.71 & 73 \\

\hline

\textbf{ABP} & Suggested \textcolor{red}{aside} the Scottish Parliamentary Constituencies Commission \textcolor{red}{indium} 2003, and adopted \textcolor{red}{atomic number 85} Holyrood on \textcolor{red}{ace Crataegus oxycantha} 2004. Area of Scotland's 32nd largest council area - covering parts of East Renfrewshire Council Area & Human-written (Succeeded) \makebox[-1pt][l]{\includegraphics[width=0.5cm]{genshin/smiley.png}} & 23.53 & 71 \\

\hline

\textbf{HQA} & Suggested by the Scottish Parliamentary Constituencies Commission in 2003, and adopted at Holyrood on 1 May 2004. Area of Scotland's \textcolor{red}{2004. sphere council} area \textcolor{red}{thirty-second prominent council} of East Renfrewshire Council Area & Human-written (Succeeded) \makebox[-1pt][l]{\includegraphics[width=0.5cm]{genshin/smiley.png}} & 20.59 & 44 \\

\hline

\textbf{T-PGD} & \textcolor{red}{Introduced. by.} Scottish Parliamentary Resituiances Commission in 2002, and adopted at Holyrood on 1 May 2004. \textcolor{red}{All of Scotland's 32 The largest councils} area - covering parts of East Renfrewshire Council Area & Human-written (Succeeded) \makebox[-1pt][l]{\includegraphics[width=0.5cm]{genshin/smiley.png}} & 26.47 & 56 \\

\hline

\textbf{\attackname} & Suggested by the \textcolor{red}{Grave} Parliamentary Constituencies Commission in 2003, and adopted at Holyrood on 1 May 2004. Area Scotland 32nd \textcolor{red}{biggest} council area - covering parts of East Renfrewshire Council Area & Human-written (Succeeded) \makebox[-1pt][l]{\includegraphics[width=0.5cm]{genshin/smiley.png}} & 5.88 & 12 \\

\hline

\end{tabular}
}
\caption{Case study of semantic preservation of the adversarial texts generated by various attack methods. Words modified during the attacks are highlighted in \textcolor{red}{red}.}
\label{tab:case_study_results}
\end{table*}

\end{document}